\documentclass[10pt]{article}\usepackage[]{graphicx}
\usepackage[usenames,dvipsnames]{xcolor}
\makeatletter
\def\maxwidth{ %
  \ifdim\Gin@nat@width>\linewidth
    \linewidth
  \else
    \Gin@nat@width
  \fi
}
\makeatother

\definecolor{fgcolor}{rgb}{0.345, 0.345, 0.345}

\newcommand{\vb}[1]{\textcolor{blue}}
\usepackage{framed}
\makeatletter
 {\par\unskip\endMakeFramed%
 \at@end@of@kframe}
\makeatother

\definecolor{shadecolor}{rgb}{.97, .97, .97}
\definecolor{messagecolor}{rgb}{0, 0, 0}
\definecolor{warningcolor}{rgb}{1, 0, 1}
\definecolor{errorcolor}{rgb}{1, 0, 0}
\definecolor{alizarin}{rgb}{0.82, 0.1, 0.26}
\newenvironment{knitrout}{}{} 
\usepackage{alltt}
\usepackage{amsmath,amsthm}
\usepackage{amssymb,graphicx,float,bbm,amssymb,fmtcount,mathtools}
\usepackage[a4paper]{geometry}
\usepackage{setspace, color}
\usepackage{bm}
\usepackage{algorithm,algpseudocode}

\usepackage{amsfonts}

\usepackage{authblk}
\usepackage[english]{babel}
\usepackage{latexsym}
\usepackage{enumerate,comment}
\usepackage{graphicx}
\usepackage{nicefrac}
\usepackage{mathrsfs}
\usepackage[round,authoryear,comma]{natbib}
\usepackage{framed}
\usepackage{multirow}
\usepackage{color}
\usepackage{epstopdf}
\usepackage{todonotes}
\usepackage{booktabs}

\usepackage{pstricks}
\def\R{{\rm I} \negthinspace {\rm R}}
\makeatother
\newtheorem{theorem}{Theorem}[section]
\newtheorem{lemma}[theorem]{Lemma}

\newtheorem{assumption}{Assumption}

\newtheorem{remark}[theorem]{Remark}
\theoremstyle{definition}
\newtheorem{definition}[theorem]{Definition}
\usepackage{natbib}

\def\bY{{\mathbf Y}} 
\def\R{{\mathbb R}}  
\def\p{{\mathbb P}}  
\def\E{{\mathbb E}}  

\def\VaR{{\textnormal{VaR}}}
\def\TVaR{{\textnormal{TVaR}}}

\newcommand{\cF}{\mathcal F}

\newcommand{\RM}{\textnormal{RM}}
\def\bY{{\mathbf Y}}
\theoremstyle{definition}
\newtheorem{examplex}{Example}
\newenvironment{example}
{\pushQED{\qed}\examplex}
{\popQED\endexamplex}

\DeclareMathOperator*{\argmin}{arg\,min}
\def\dTVaR{{\textnormal{dTVaR}}}

\usepackage{natbib}

\usepackage{setspace}

\title{\bf Two step valuations}
\IfFileExists{upquote.sty}{\usepackage{upquote}}{}
\usepackage{hyperref}

\begin{document}
	
\title{\LARGE \textbf{Insurance valuation: A two-step generalised regression approach}}
\date{Version: \today }
\author[1]{\large Karim Barigou\footnote{E-mail address: \href{mailto:karim.barigou@univ-lyon1.fr}{karim.barigou@univ-lyon1.fr} (Corresponding author)}}
\author[2]{\large Valeria Bignozzi\footnote{E-mail address: \href{mailto:valeria.bignozzi@unimib.it}{valeria.bignozzi@unimib.it}}}
\author[3]{\large Andreas Tsanakas\footnote{E-mail address: \href{mailto:a.tsanakas.1@city.ac.uk}{a.tsanakas.1@city.ac.uk}}} 
\affil[1]{\footnotesize ISFA, Univ Lyon 1, UCBL, LSAF EA2429, F-69007, Lyon, France}
\affil[2]{\footnotesize Department of Statistics and Quantitative Methods, University of Milano-Bicocca, Italy}
\affil[3]{\footnotesize {Bayes Business School}, City, University of London, UK} 
\maketitle

\begin{abstract}
	
	Current approaches to fair valuation in insurance often follow a two-step approach, combining quadratic hedging with application of a risk measure on the residual liability, to obtain a cost-of-capital margin. In such approaches, the preferences represented by the regulatory risk measure are not reflected in the hedging process.  We address this issue by an alternative two-step hedging procedure, based on generalised regression arguments, which leads to portfolios that are neutral with respect to a risk measure, such as Value-at-Risk  or the expectile.  First, a portfolio of traded assets aimed at replicating the liability is determined by local quadratic hedging. Second, the residual liability is hedged using an alternative objective function. The risk margin is then defined as the cost of the capital required to hedge the residual liability. In the case quantile regression is used in the second step, yearly solvency constraints are naturally satisfied; furthermore, the portfolio is a risk minimiser among all hedging portfolios that satisfy such constraints. We present a neural network algorithm for the valuation and hedging of insurance liabilities based on a backward iterations scheme. The algorithm is fairly general and easily applicable, as it only requires simulated paths of risk drivers.
	
	\textbf{Keywords}: Market-consistent valuation, Quantile regression, Solvency II, Cost-of-capital, Dynamic risk measurement.
	
\end{abstract}

\newpage

\section{Introduction}
 
 Fair valuation of insurance liabilities has become a fundamental feature of
 modern solvency regulations in the insurance industry, such as the Swiss
 Solvency Test, Solvency~II and C-ROSS (Chinese solvency regulation), see e.g. \cite{solvency2009directive}. Broadly speaking, insurance regulations distinguish between liabilities that are completely replicable in deep, liquid and transparent markets and liabilities for which this is not possible. In the first case, by no-arbitrage arguments, the fair value should correspond to the initial cost of the replicating portfolio. Otherwise, the fair value is defined as the sum of the expected present value (called best-estimate) and a risk margin that is based on cost-of-capital arguments. 
  
 When insurance liabilities are  a combination of traded and non-traded risks, as is often the case, one cannot classify them as perfectly replicable or non-replicable. In this context, it is not  evident how the regulatory valuation should proceed and the valuation is therefore usually performed in two steps. In a first step, a hedging portfolio for the liabilities is set up, based on the available traded assets via typically a quadratic objective function \citep{wuthrich2013financial,pelsser2016difference,wuthrich_2016}. In a second step, a risk margin is added to account for the residual risk via a risk measure or an actuarial principle \citep{embrechts2000,happ2015best,dhaene2017fair}. \textcolor{black}{\cite{mohr2011market} proposed a valuation framework based on replication over multiple one-year
 	time periods by a periodically updated portfolio of assets. In that framework, the split of the total asset portfolio into the value of insurance liabilities and capital funds is based on an acceptability condition related to the expected return for capital investor (see also \citealp{engsner2017insurance,engsner2020value}).} 
  
 As far as the hedging procedures are concerned, several objective functions have been proposed in the literature. Some papers considered maximising the expected utility of the hedger (\citealp{henderson2004utility}), using indifference arguments (\citealp{moller2003indifference}) or minimising the risk by quadratic hedging (\citealp{schweizer1995variance,moller2001hedging,moller2001risk,delong2019fair2}). The major drawback of using a quadratic penalty function is that it penalises equally gains and losses. Furthermore, quadratic hedging, leading to residuals with zero expectation, is divorced from the preferences encoded in regulation, which require neutralisation of a risk measure such as Value-at-Risk (VaR). \cite{follmer2000efficient} and \cite{francois2014optimal} proposed to use general expected penalties that only penalise losses. \cite{follmer1999quantile} defined the quantile hedging scheme, which maximises the probability that the hedging loss does not exceed a certain threshold at maturity, given an initial capital. In our paper, by \textit{quantile hedging}, we mean hedging with a quantile regressor as considered in \cite{koenker1978regression} which is a related but different objective compared to the quantile hedging of \cite{follmer1999quantile}.
 
\textcolor{black}{ We propose a valuation approach for insurance liabilities, which addresses the following considerations:
 \begin{enumerate}
 \item We consider hedging strategies that produce a residual which has zero tail risk, as measured by a VaR or Expectile criterion. Hence the trading strategy is directly tied to the regulatory criterion that insurers have to satisfy and does not artificially restrict them to invest capital in risk-free assets only, if the option of a more risk-sensitive investment exists.
   \item The use of such  trading strategies mitigates potential regulatory arbitrage arising from the use of VaR. It is known that the requirement for VaR-neutrality of portfolios can be achieved by shifting risk to the extreme tail (e.g. \citealp{danielsson2002emperor}). We avoid such perverse incentives by the requirement of minimising a convex loss function; hence we are not only concerned with the zero-VaR residual property, but also with how this is achieved.
 \item Fair values are generated by apportioning the hedging cost to policyholders and shareholders. \textcolor{black}{Specifically, we assume that shareholders contribute a part of the hedging cost, interpreted as risk capital, while policyholders contribute a part  that covers liabilities on average and, in addition, compensate shareholders for their cost of capital. } The resulting fair values generalise current regulatory valuation formulas. Hence, rather than introducing a radically new approach, we build on current practice, allowing for more risk-sensitive trading strategies. 
\item \textcolor{black}{In a multi-period setting, we only consider the case of a terminal liability, and not of cash flows of payments before maturity. We make the (strong) assumption that at intermediate times when the portfolio is re-balanced, insurers are able to raise capital from shareholders as needed, with respective capital costs reflected as part of the valuation.}
  \end{enumerate}}

First, we introduce a new valuation framework for the multi-period fair valuation of insurance liabilities based on a two-step hedging procedure. The framework we present makes use of sequential local quadratic and quantile risk-minimising strategies to take into account all intermediate solvency requirements. By considering a local approach, not only we focus on the loss at maturity but also on the difference between the hedging portfolio and the liability value at intermediate times (for instance on a yearly basis) which is of paramount importance in a regulatory context given yearly solvency constraints. Moreover, by switching from a quadratic to an alternative loss function, we target the tail risk rather than the mean of the residual loss and therefore account explicitly, as part of the hedging and valuation process, for those extreme events that drive capital requirements.
 
 The two-step approach can be summarised as follows. In a first step, a portfolio of traded assets aimed at replicating the liability as closely as possible (in the quadratic sense) is determined similar to \cite{follmer1988hedging}. Such portfolio replicates the liability \textit{on average}. In a second step, the residual liability is managed by generalised error functions, which are associated with setting different statistics of residual loss (e.g. VaR or Expectile) to zero. In particular, our focus is on local quantile hedging using the asymmetric Koenker-Basset error (cf. \citealp{koenker1978regression}) given the use of VaR in regulation. The resulting hedging portfolio appears as a Tail Value-at-Risk  deviation risk minimiser among all portfolios which satisfy VaR constraints (\citealp{rockafellar2008risk}). Hence, one can achieve a better risk management by quantile-hedging the residual risk rather than setting up a VaR capital buffer (cf. Lemma \ref{lemmatvar}). 
 
 Second, the fair value is then defined as the sum of the cost of the quadratic hedging portfolio and the cost-of-capital for the quantile hedging of the residual liability. \textcolor{black}{This construction of fair values reflects an economic framework whereby an insurer's shareholders contribute the risk capital, understood as the cost of the quantile hedging of the second step, while policyholders contribute the cost of the quadratic hedging strategy of the first step plus a compensation for shareholders' capital investment. Hence, policyholder contributions ensure that the portfolio is neutral in expectation, in line with standard provisioning conventions, while shareholder capital ensures VaR-neutrality, thus reflecting regulatory criteria.} \textcolor{black}{The policyholder contribution ensures that shareholders, who capitalise the insurer to a level prescribed by regulation, are thus compensated for the risk they bear. In this way, the proposed approach to some extent integrates regulatory and economic aspects of valuation.}
  
 Third, we propose a general algorithm for the valuation and hedging of insurance liabilities based on a backward iteration scheme. Standard implementation of local hedging strategies requires dynamic programming in discrete time, which leads to high computational times (see \citealp{cerny2004dynamic} and \citealp{augustyniak2017assessing}). Some recent papers proposed machine-learning based algorithms to speed up the global hedging problem (\citealp{fecamp2019risk,carbonneau2020equal,carbonneau2021deep}) but we are not aware of similar algorithms for local quantile hedging. Contrary to the standard dynamic programming approach, our algorithm does not present a nested structure and only requires sample paths of the main risk drivers. This is of practical importance, as typically the stochastic asset models used in insurance do not have simple tractable forms; instead insurers have access to the output of Economic Scenario Generators (\citealp{varnell2011economic}), which provide precisely a matrix of sample paths for multiple asset classes. In this paper, we focus on a neural network implementation for quantile hedging, but the algorithm remains valid with other non-linear regression and loss functions. Furthermore, this paper provides a practical implementation to generate future paths for the fair values and their corresponding hedging portfolios.
 
 The paper is organised as follows. Section \ref{section2} starts by motivating our two-step valuation approach in a one-period setting and its connection to the Solvency II regulatory framework. Moreover, we show that the two-step valuation is market-consistent and actuarial in the sense of \cite{dhaene2017fair}. In Section \ref{sectionmultiperiod}, we generalise the valuation approach in a dynamic multi-period setting by sequential risk-minimisation. Section \ref{sectionneuralnetwork} presents a general procedure for implementing the dynamic hedging problem and proposes a neural network algorithm based on Monte-Carlo simulations of the financial and actuarial risk drivers. Section \ref{numericalexamples} provides a detailed numerical example illustrating the neural network approximation. Finally, brief
 conclusions are given in Section \ref{conclusion}.

   \section{Fair valuation in a one-period setting} \label{section2} 
  
 We start by investigating the one-period case. Let $(\Omega,\mathcal{F}, \mathbb{P})$ be a probability space and denote $\mathcal{C}\subseteq L^2(\Omega,\mathcal{F},\mathbb{P})$ the set of all claims with maturity $T=1$.  We assume that the financial market consists of    asset $0$, which is risk-free with   deterministic interest rate $r\geq 0$, and  $n$   risky assets. The vectors $\mathbf{y}=(1,y_1,\ldots,y_n)$ and  $\mathbf{Y}=\left(Y_0,Y_1,\dots, Y_n\right)$ respectively represent the asset values  at time $0$ and time $1$ where $Y_0=e^r$, $y_i\geq 0,~Y_i\in\mathcal{C}$ for any $i=1,2,\ldots,n$. A trading strategy $\bm \beta=(\beta_0,\ldots,\beta_n)$ is a real-valued vector where $\beta_i$ provides the  units of capital invested in asset $i$  at time $0$. We assume that the strategy is not modified over time and denote $\mathcal{B}={\mathbb{R}^{n+1}}$  the set of all trading strategies. The value of the trading strategy at time $0$ and $1$ is obtained as
 $$
\bm \beta\cdot\mathbf{y}=\sum_{i=0}^n\theta_i\cdot y_i \quad\textrm{and}\quad\bm \beta\cdot\mathbf{Y}=\sum_{i=0}^n\theta_i\cdot Y_i. 
 $$
 
 We further assume that all the assets are non redundant, i.e. $\bm \beta\cdot \mathbf{Y}=0$ if and only if $\bm \beta=(0,\ldots,0)$\footnote{This assumption ensures that the quadratic minimisation problem has a unique solution, see e.g. \cite{cerny2009hedging}} and  that any tradable asset can be bought and/or sold in any quantity in a deep, liquid and transparent market with negligible transactions costs and other market frictions. All inequalities between random variables  are understood to hold $\mathbb{P}$-almost surely.
  
  \subsection{Two-step valuation with a quadratic loss function} Solvency regulations require  a \textit{fair valuation} of assets and liabilities, that is, their value should correspond to the amount for which they could be transferred to another company or exchanged on the market.\footnote{Article 75 in the Solvency II Directive: ``Assets and liabilities shall be valued at the amount for which they could be transferred, or settled, between knowledgeable willing parties in an arm’s length transaction.'' As pointed out by a referee, such ``transfer valuation'' is not the only plausible criterion, as one could focus on the fulfillment rather than the transfer of liabilities.} For this reason the valuation of a contingent claim  strictly depends on whether it is tradable on  the financial market. 
  
  Similar to \cite{dhaene2017fair}, we denote by $\mathcal{C}^h\subseteq \mathcal{C}$ the class of claims perfectly hedgeable on the market. For any $S\in \mathcal{C}^h$ it is possible to find a strategy $\bm \beta\in\mathcal{B}$ such that  $S=\bm \beta\cdot \mathbf{Y}$. In this case, the fair value of the liability $S$ is simply given by the value of the trading strategy at time $0$, $\bm \beta\cdot\mathbf{y}$.
  
Moreover, we denote by $\mathcal{C}^\perp\subseteq \mathcal{C}$ the class of claims independent of the vector  of traded assets $\mathbf{Y}=\left(Y_0,Y_1,\dots, Y_n\right)$. For such claims, the position of the insurer cannot be hedged in the financial market and therefore
  the fair value of $S\in\mathcal{C}^\perp$ is calculated by an actuarial premium principle. In a solvency framework, with capital requirements calculated according to the Value-at-Risk ($\VaR$) risk measure, the standard choice is the cost-of-capital premium principle, see for instance \cite{solvency2009directive} and \cite{pelsser2011}. It is defined as follows:
\begin{equation}\label{CoC PP}
	\pi(S)=e^{-r}\ \mathbb{E}\left[S\right]+e^{-r}~  i~(\VaR_\alpha (S)-\mathbb{E}\left[S\right]),
\end{equation}
where $i\in(0,1)$ is the cost-of-capital rate and  
$$
\VaR_\alpha(X)=\inf\{x\in\R~|\p(X\leq x)\geq \alpha\}, \quad \textrm{for any risk } X{\in\mathcal{C}}~\textrm{and  } \alpha\in(0,1).$$ Thus, $\pi(S)$ is understood as the expected value (net premium) of $S$, loaded by the cost of the capital required to hold the liability $S-\mathbb{E}[S]$. Specifically, it is assumed that the insurer's shareholders require a return $i$ on $\VaR_\alpha (S-\mathbb{E}\left[S\right])$, that is, on the assets required to support $S$, net of the expected value. We note that the premium principle \eqref{CoC PP} is not appropriate for claims that depend on traded assets as this would neglect the hedging opportunities.

Many claims that insurance companies face are not perfectly hedgeable, but nevertheless not independent of the payoffs of the traded assets. We call these claims \textit{hybrid claims} when $S\in \mathcal{C}\backslash(\mathcal{C}^h\cup \mathcal{C}^\perp)$ and these are the focus of our paper.

 In that case, some (generally imperfect) hedging of $S$ by $\mathbf Y$ is possible and typically a two-step approach is followed (see \citealp{mohr2011market} and \citealp{albrecher2018asset}). First, the insurer determines a hedging portfolio that is as close as possible to the liability $S$. To measure ``closeness'' {the quadratic loss function is generally used}, providing a trading strategy \textcolor{black}{$\bm \theta_S=(\theta_0,\ldots,\theta_n)$} that minimises the $L^2$-distance between the liability and the hedging portfolio (see \citealp{pelsser2016difference}): 
 \begin{equation}
\bm \theta_S=\arg \min_{\bm \beta\in\mathcal{B}} \mathbb{E}\left[(S-\bm \beta \cdot \mathbf{Y})^2\right].\label{quadratic}
 \end{equation}
 For the rest of the paper $\bm\theta_S$ will always denote the trading strategy associated with the quadratic loss function and the index $S$ will be dropped when no confusion is possible. Standard least-squares arguments provide the solution to problem \eqref{quadratic}, $\bm \theta=(\mathbb{E}\left[\mathbf{Y}^\intercal \cdot\mathbf{Y}\right])^{-1}\mathbb{E}\left[S\cdot\mathbf{Y}^\intercal\right]
$\footnote{\textcolor{black}{The vector $\mathbf{Y}^\intercal$ represents the transpose of the asset price vector $\mathbf{Y}$ and the non-redundancy guarantees the existence of the inverse, see Theorem 5 in \cite{dhaene2017fair} for more details.}} and ensure that the expected value of the hedging portfolio matches the expected value of the liability: 
\begin{equation}
	\mathbb{E}\left[\bm \theta \cdot \mathbf{Y}\right]=	\mathbb{E}\left[S\right].\label{expproperty}
\end{equation}

Second, the insurer values the residual risk $R(S,\bm \theta):=S-\bm \theta \cdot \mathbf{Y}$, which could not be hedged, by the cost\textcolor{red}{-}of\textcolor{red}{-}capital principle $\pi$. Following such a method, the fair value $\phi(S)$ is calculated as the sum of the cost of the hedging portfolio and the premium principle of the residual risk (see \citealp{dhaene2017fair}):
\begin{align}
\phi(S)&=\bm \theta\cdot \mathbf{y}+\pi (R(S,\bm \theta)) \nonumber\\
&=\bm \theta\cdot \mathbf{y}+e^{-r} i \VaR_\alpha(S-\bm \theta\cdot \mathbf{Y}),\label{rhococ}
\end{align}
where we used the property \eqref{expproperty}. 

The fair value \eqref{rhococ} can be seen as a generalisation of the premium principle \eqref{CoC PP}, where (a) the net premium $\mathbb E[S]$ is replaced by the cost of the hedging strategy $\bm \theta\cdot \mathbf{y}$, which once again matches on average the liability $S$ and (b) the cost of capital is calculated on the residual $S-\bm \theta\cdot \mathbf{Y}$, rather than $S-\mathbb E[S]$.

\subsection{Valuation with general loss functions}

The valuation approach we just discussed relies on the use of a quadratic loss function $\ell(x)=x^2$ to penalise deviations of the hedging portfolio payoff from the liability. Here, we generalise the two-step valuation approach, addressing two specific concerns:
\begin{itemize}
    \item The quadratic loss function penalises losses and gains equally. As an insurer, the major concern is to avoid a shortfall, namely situations where $S>\bm \theta \cdot \mathbf{Y}$. Various authors (for instance \citealp{follmer2000efficient} and \citealp{francois2014optimal}) proposed alternative penalty functions that only penalise losses or penalise losses and gains asymmetrically.
    \item The total level of assets that the insurer has to hold with respect to their liabilities is typically given by a risk measure, e.g. in the case of Solvency II, $\VaR_{0.995}$. It is then not obvious why a quadratic hedging strategy should be used, which results in a residual risk with mean zero, rather than a hedging strategy that produces a $\VaR_{0.995}$-neutral portfolio.
\end{itemize}

To elaborate on these points, consider a convex loss function $\ell:\R\to[0,+\infty)$, \textcolor{black}{with $\ell(x)=0$  if and only if $x=0$.} The resulting hedging strategy \textcolor{black}{$\bm \xi^{(\ell)}_S=(\xi^{(\ell)}_0,\ldots,\xi^{(\ell)}_n)$ is defined as the minimiser of the expected loss function}
 \begin{equation}
{\bm \xi^{(\ell)}_S}=\arg \min_{\bm \beta\in\mathcal{B}} \mathbb{E}\left[\ell(S-\bm \beta \cdot \mathbf{Y})\right]. \label{regr general}
 \end{equation}
Again, we drop the subscript to write ${\bm \xi^{(\ell)}}$, if no confusion ensues. Different choices of the loss function lead to the risk-neutrality (or unbiasedness) of the residual risk with respect to different risk measures. 
Specifically, under mild conditions  we have that \citep[see Thm 3.2 in][]{rockafellar2008risk}
$$
\Gamma^{(\ell)}(S-{\bm \xi^{(\ell)}}\cdot \mathbf{Y})=0,
$$
where $\Gamma^{(\ell)}$ is the risk measure given by 
\begin{equation}\label{def:shortfall}
\Gamma^{(\ell)}(X)=\arg \min_{c \in \mathbb R } \mathbb{E}\left[\ell(X-c)\right],\quad \textrm{for any }X\in\mathcal{C}. \end{equation}
Slightly different versions of risk functionals as defined in \eqref{def:shortfall} are treated in the  literature using a first order condition to find the minimiser in \eqref{regr general}, see for instance the \textit{zero-utility premium principles} discussed in \cite{deprez1985convex}, the class of \textit{shortfall risk measures}
introduced by \cite{follmer2002convex}, the \textit{generalised quantiles} investigated by \cite{bellini2014generalized}, the {optimised certainty equivalent} in \cite{ben2007old}, the \textit{elicitable functionals} studied in \cite{gneiting2011making} and more recently the class of \textit{convex hedgers} by \cite{dhaene2017fair}. The implications of different choices of $\ell$ are elaborated on in detail by \cite{rockafellar2013quadrangle}. Three important examples are:
\begin{enumerate}
    \item If $\ell(x)=x^2$, we have that $\Gamma^{(\ell)}(X)=\mathbb E[X]$ and ${\bm \xi^{(\ell)}}= \bm\theta$, as seen before.
    
    \item Let $\ell(x):= \ell_{\alpha}(x)$, where
    \begin{equation}
     \ell_{\alpha}(x)=\frac{\alpha}{1-\alpha}x_++x_-,\qquad \alpha\in(0,1)\label{penfunction}
    \end{equation}
with \begin{align*}
x_{+}&=\max(x,0)\\
x_{-}&=\max(-x,0),
\end{align*}
    is the normalised Koenker-Bassett loss function, see \cite{Koenker2005}. Then, $\Gamma^{( \ell_{\alpha})}(X)=\mathbb \VaR_\alpha(X)$. We henceforth denote the trading strategy associated with this loss function by ${\bm \xi^{(\ell_\alpha)}}:= \bm \xi$. This strategy satisfies 
  \begin{equation}\label{VaRneutrality}
\VaR_\alpha \left(S-\bm \xi \cdot \mathbf{Y}\right)=0.
\end{equation}
\noindent This case is important to us, given the desired feature that the VaR of the residual risk is zero, indicating that sufficient assets have been allocated to satisfy regulatory requirements. 
    
    \item Alternatively, with slight abuse of notation, consider the loss function $\ell(x):=\ell_\tau(x)$ where 
    \begin{equation*}
    \ell_{\tau}(x)=\tau (x_+)^2 + (1-\tau)(x_-)^2, \quad \tau \in (0,1).\label{expectileloss}
    \end{equation*}
    The resulting risk measure \textcolor{black}{$$\Gamma^{(\ell_\tau)}(X)=\arg \min_{c \in \mathbb R } \mathbb{E}\left[\tau ((X-c)_+)^2 + (1-\tau)((X-c)_-)^2\right]$$}is the $\tau$-\textit{expectile}. Expectiles generalise the mean (which is obtained by setting $\tau=0.5$) and, for $\tau\geq 0.5$, are coherent risk measures, thus addressing a common criticism of VaR, while remaining within the tractable class of shortfall risk measures \citep[see for instance][]{delbaen2016risk}. One can see hedging with $\ell_\tau(x)$ as a modification of quadratic hedging, where, for $\tau>0.5$ additional weight is given to the downside risk.  Expectile regression was introduced
    by \cite{NeweyPowell} and then further generalised to $M$-quantiles by \cite{breckling1988m}. We return to \textit{expectile hedging strategies} in Section \ref{sectiontwostep1period}.
\end{enumerate}

\textcolor{black}{By changing the hedging objective from a
quadratic to an asymmetric loss function we move the focus on  positive deviations of the hedging portfolio from the liability.  As a result, the residual risk changes from
having zero mean to having a zero VaR or Expectile. Specifically, by \eqref{VaRneutrality}, if we set up the strategy $\bm \xi$ for $\alpha=0.995$, by construction, the portfolio will cover the liability $S$ with probability $\alpha=0.995$. In this paper, we focus on the use of the quantile hedging strategy 
\begin{equation}
    \bm \xi = \arg \min_{\bm \beta\in\mathcal{B}} \mathbb{E}\left[ \ell_{\alpha}(S-\bm \beta \cdot \mathbf{Y})\right], \label{quantilehedging}
\end{equation}
where $\ell_{\alpha}(x)$ is given in \eqref{penfunction}.} The regression problem \eqref{quantilehedging} is well-known as the quantile regression pioneered by \cite{koenker1978regression}. As we use it for hedging objectives, we refer to this minimisation as \textit{quantile hedging}; this should not be confused with the quantile hedging of \cite{follmer1999quantile} which targets a different objective. 

Moreover, as the quantile hedging strategy is a risk minimiser with respect to a convex loss function, we will avoid situations where capital requirements can be reduced by shifting losses to the extreme tails, beyond the VaR level. Hence one of the key criticisms of VaR, see e.g. Section 4.4 of \cite{danielsson2002emperor}, is addressed, as  illustrated in the following example.

\begin{example}\label{example: regulatory arbitrage}
We consider a simple example, where there is only one risky asset  correlated with the liability $S$. Assume that the risk measure used is $\VaR_{0.9}$ and that the risk-free asset  has return 1 (zero interest rates, in particular $y_0=Y_0=1$). The liability $S$ follows a Lognormal distribution with parameters $\mu=0.1$ and $\sigma=0.3$, such that  $\VaR_{0.9}(S)=1.623$. The asset $Y_1$ is a derivative on $S$, with a price of $y_1=1$ and pay-off:
$$
Y_1= 1.5 \cdot \mathbf 1_{\{S\leq \VaR_{0.9}(S)\}} - 3 \cdot \mathbf 1_{\{S> \VaR_{0.9}(S)\}}.
$$
Hence, the derivative offers a high return when $S$ is less than $\VaR_{0.9}(S)$, but produces an even larger loss when $S$ exceeds $\VaR_{0.9}(S)$. Investment in such an asset would be capital efficient, as it moves the loss beyond the 90\% confidence level, thus making it `invisible' to $\VaR_{0.9}$. On the other hand, prudent risk management would require the holder of $S$ to \emph{short} $Y_1$, in order to be able to hedge their tail risk. 

To make these considerations precise, we look at two investment strategies:
\begin{itemize}
	\item[$\bm A$] Invest $\beta_0= 0$ in the risk-free asset $Y_0$ and $\beta_1=\VaR_{0.9}(S)/1.5=1.082$ in $Y_1$. The resulting portfolio Value-at-Risk is:
	\begin{align*}
	\VaR_{0.9}(S-\beta_1Y_1)&=\VaR_{0.9}\left(S-\VaR_{0.9}(S)\left(\mathbf 1_{\{S\leq \VaR_{0.9}(S)\}} - 2 \cdot \mathbf 1_{\{S> \VaR_{0.9}(S)\}}\right)\right)\\
	&= \VaR_{0.9}\left(S-\VaR_{0.9}(S)+3\VaR_{0.9}(S) \mathbf 1_{\{S> \VaR_{0.9}(S)\}}\right)\\
	&=\VaR_{0.9}\left(S-\VaR_{0.9}(S)\right)+3\VaR_{0.9}(S) \VaR_{0.9}\left(\mathbf 1_{\{S> \VaR_{0.9}(S)\}}\right)\\
	&=0,
	\end{align*}
	where the 3rd equality is by comonotonic additivity of VaR.
	\item[$\bm B$] Invest $\bm\xi =(\xi_0,\xi_1)$ in the assets $(Y_0, Y_1)$ , where $\bm\xi$ is the quantile hedging strategy at level $\alpha=0.9$. By construction we have
	\begin{align*}
	\VaR_{0.9}(S-\xi_0-\xi_1Y_1)&=0.
	\end{align*}
	The corresponding optimal positions are 
	\begin{align*}
	\xi_0&=1.697,\\
	\xi_1&=-0.174.
	\end{align*}
	Hence, it is indeed seen that a negative exposure to $Y_1$ is produced. Furthermore, the \emph{cost} of this strategy is equal to $\xi_0+\xi_1=1.523 >1.082=\beta_1$. Hence the Strategy B is more expensive, while achieving the same zero VaR for the portfolio as Strategy A.
\end{itemize}

We compare these two strategies with respect to the cumulative distribution function (cdf) of their residual risks in Figure \ref{fig:regulatory arbitrage}. The persistent tail risk arising from Strategy A is clearly visible, indicating how this strategy, while cost efficient, reflects poor (unethical even) risk management. This is not a problem we face with the quantile hedging Strategy B.

\begin{figure}
    \centering
    \includegraphics[width=\maxwidth]{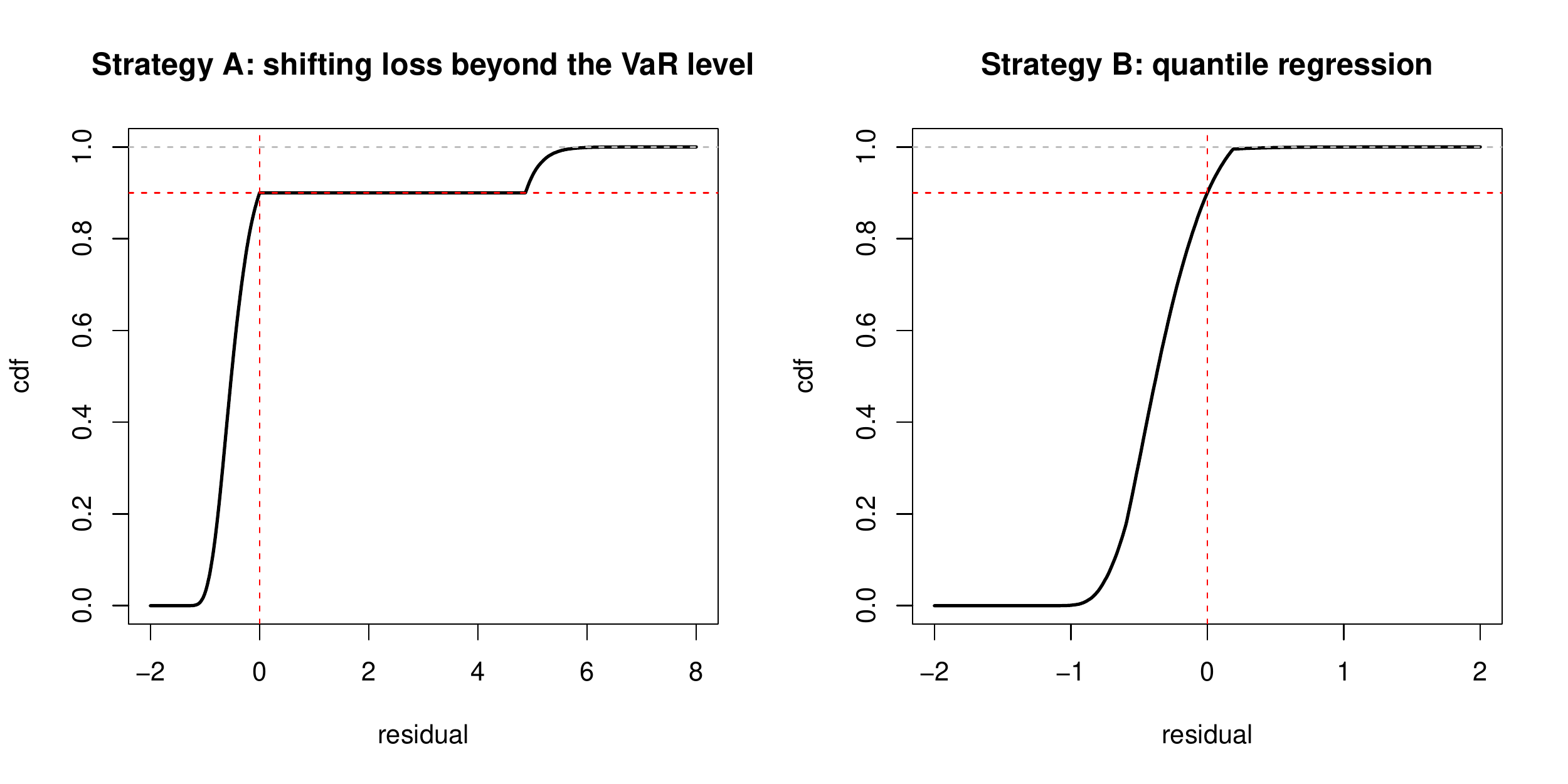} 
    \caption{Cumulative distribution functions of the residual risks $R(S,\bm \beta )$ (left) and $R(S,\bm \xi )$ (right), corresponding, respectively, to the Strategies A and B of Example \ref{example: regulatory arbitrage}.}
    \label{fig:regulatory arbitrage}
\end{figure}

\end{example}

\subsection{Two-step valuation with quadratic and quantile hedging} \label{sectiontwostep1period}

\textcolor{black}{We have considered different loss functions  to yield a portfolio that is more suitable for fair valuation in a solvency context, compared to quadratic hedging. Valuation still needs to take into account that only a fraction of  capital requirements is borne by policyholders. Here we show that the hedging strategy based on a general loss function can be obtained also by a two-step  approach, where quadratic hedging is used as a first step and  hedging based on a convex loss function is subsequently applied on the residual liability (see Lemma \ref{hedgingdecomposition} for more details).} \textcolor{black}{The two-step approach is designed based on the idea that assets can be raised from a capital investor, if sufficient return can be provided. Hence the value of the liabilities is not equal to the full cost of the hedging strategy $\bm \xi$, but needs instead to cover shareholders' capital costs for funding this strategy.}
\textcolor{black}{Hence, quantile hedging ensures that regulatory requirements are satisfied, while the decomposition of the hedging costs to a part that is fully borne by policyholders and a `cost-of-capital' part, reflects the source of funding of the strategy. As a result we obtain a generalisation of \eqref{rhococ}.}

\textcolor{black}{Before discussing in depth this approach, we briefly  recall the definition of a valuation and its properties as introduced in \cite{dhaene2017fair}.}
\begin{definition}[Valuation] \label{def:valuation}A valuation is a mapping
\begin{equation*}
\rho~:~\mathcal{C}\to \R,~S\mapsto \rho(S),
\end{equation*}
that is \textit{normalised}
$\rho(0)=0$
and \textit{translation invariant} $$\rho(S+a)=\rho(S)+e^{-r}a,\quad \forall S\in\mathcal{C},~a\in\R.$$

\end{definition}
Here we list some properties that a valuation  may satisfy and which will be discussed in the following. For any $S,S_1,S_2\in\mathcal{C}$ we say that a valuation $\rho$ is
\begin{enumerate}
\item \textit{Positive homogeneous} if 
$$
\rho(\lambda S)=\lambda \rho(S),\qquad \textrm{for any } ~\lambda\geq 0;
$$
\item \textit{Market-Consistent} if 
$$
\rho(S+S^h)=\rho(S)+\bm\beta\cdot \mathbf{y},
$$
 for any \textit{hedgeable} payoff: $S^h=\bm \beta\cdot \mathbf{Y}$;
\item \textit{Actuarial} if \begin{equation*}
\rho (S^\perp)=e^{-r} \mathbb{E}[S^\perp]+\RM(S^\perp), 
\end{equation*}
for any claim $S^\perp \in\mathcal{C}^\perp $,  where $\RM:\mathcal{C}^\perp\to \R$ is a mapping that does not depend on current asset prices $\mathbf{y}$;
\item \textit{Fair} if $\rho$  is market-consistent and actuarial.
\end{enumerate}
 The market-consistency property means that the valuation is marked-to-market for any hedgeable part of a liability, while a valuation is actuarial if it is marked-to-model for any claim which is independent of the financial market. {The definition of \textit{fair valuation} was recently introduced by \cite{dhaene2017fair} to formalize the valuation of hybrid claims as required by  solvency regulation (see also \citealp{barigou2019fair,delong2019fair})}.
 
\begin{definition}[Two-step valuation]\label{deftwostep}
Consider the liability $S$ \textcolor{black}{and a  convex loss function $\ell:\mathbb{R}\to [0,+\infty)$, such that $\ell(x)=0$ if and only if $x=0$}. Let $\bm \theta$ be the quadratic hedging strategy for $S$ and let \textcolor{black}{$\bm \eta^{(\ell)}=(\eta^{(\ell)}_0, ,\ldots,\eta^{(\ell)}_n)$ be the hedging strategy for the residual risk $R(S,\bm\theta)=S-\bm \theta \cdot \mathbf{Y}$ with  loss function $\ell$, that is, the strategy that minimises the expected loss}:
\begin{equation}\bm \eta^{(\ell)}:=\arg \min_{\bm \beta\in\mathcal{B}} \mathbb{E}\left[ \ell \big(R(S,\bm\theta)-\bm \beta \cdot \mathbf{Y}\big)\right].\label{eta}
\end{equation}
 Then, the \emph{two-step valuation}\footnote{We remark that our two-step approach should not be confused with the \textit{two-step market valuation} proposed by \cite{pelsser2014time} and the two-step evaluators of \cite{assa2018market} which have both different meanings.} for $S$ is defined as
\begin{equation}
\rho^{(\ell)}(S):= \bm \theta\cdot \mathbf{y} + i \ \bm \eta^{(\ell)} \cdot \mathbf{y}.\label{twostepdef}
\end{equation}
\end{definition}
\begin{remark}\textcolor{black}{The assumption that $\ell(x)=0$ for $x=0$ only, reflects our concern with using the function $\ell$ to construct a ``closeness'' criterion between investment return and liability. While we are interested in penalising the downside more than the upside, we still consider that there is a cost associated with the case $R(S,\bm \theta)<\bm\beta\cdot\bm Y$ in \eqref{eta}. This excludes functions such as $\ell(x)=\max(x,0)$, which are concerned with the downside only and are used for instance in \cite{follmer2000efficient}}
\end{remark}
\textcolor{black}{Using the translation invariance property of the ordinary least square regression, it is immediate to verify that the two-step valuation is indeed a valuation according to Definition \ref{def:valuation}. The following lemma shows that performing convex hedging on the liability or performing quadratic hedging on the liability and then convex hedging on the residual risk leads to the same hedging strategy. Hence our two-step approach is consistent with constructing an asset portfolio that ``convex-hedges'' the liability $S.$ We recall from (\ref {regr general}) that $\bm\xi^{(\ell)}_S$ is the hedging strategy that minimises the expected loss function for the liability $S$ ($
{\bm \xi^{(\ell)}}=\arg \min_{\bm \beta\in\mathcal{B}} \mathbb{E}\left[\ell(S-\bm \beta \cdot \mathbf{Y})\right]
$), while $\bm \eta^{(\ell)}$ is the hedging strategy for the same expected loss function applied to the residual risk as defined in \eqref{eta}.}

\textcolor{black}{\begin{lemma}\label{hedgingdecomposition}
\begin{enumerate}\item[a)]If $\bm \xi^{(\ell)}$ is the unique minimiser of the expected loss for the function $\ell$ and liability $S$ as defined in \eqref{regr general}, and if $\bm \theta$ and $ \bm \eta^{(\ell)}$ are the hedging strategies as defined in \eqref{quadratic} and \eqref{eta}, then we have
\begin{equation*}
 \bm \xi^{(\ell)}=\bm \theta+\bm \eta^{(\ell)}.
\end{equation*}
\item[b)]If the  hedging strategy for $S$ is not unique, denote by $\{\bm \xi^{(\ell),j}\}_{j\in\mathcal{A}} $ the  set of minimisers of \eqref{regr general}. Then, for any $j\in\mathcal{A}$, we have that
\begin{equation*}
\bm \xi^{(\ell),j}=\bm \theta+\bm \eta^{(\ell),j},
\end{equation*}
where $\bm \eta^{(\ell),j}\in\{\bm \eta^{(\ell),j}\}_{j\in\mathcal{A}} $, is the set of  hedging strategies for the residual risk $S-\bm \theta\cdot \mathbf{Y}$.\end{enumerate}
\end{lemma}}
\textcolor{black}{\begin{proof}
	\begin{enumerate}\item[a)]Let us first consider the case where the hedging strategy is unique. The proof is by contradiction. Assume that $ \bm \xi^{(\ell)}\neq \bm \theta+\bm \eta^{(\ell)}$ and define $\bm\eta^{{(\ell)},*}=\bm\xi^{(\ell)}-\bm\theta$. By \eqref{regr general} we have 
	$$
	 \mathbb{E}\left[ \ell(S-\bm \xi^{(\ell)}_S \cdot \mathbf{Y})\right]=\mathbb{E}\left[\ell(S-(\bm \theta+\bm \eta^{(\ell),*} ) \cdot \mathbf{Y}\right)]<  \mathbb{E}\left[ \ell(S-(\bm \theta+\bm \eta^{(\ell)} ) \cdot \mathbf{Y})\right],
	$$
	which contradicts the definition of $\bm \eta^{(\ell)}$:
	$$
	\bm \eta^{(\ell)}=\arg \min_{\bm \beta\in\mathcal{B}} \mathbb{E}\left[ \ell(S-\bm \theta \cdot \mathbf{Y}-\bm \beta \cdot \mathbf{Y})\right].
	$$
\item[b)]For the non-uniqueness case, take $\bm \xi^{(\ell),j}\in\{\bm \xi^{(\ell),j}_S\}_{j\in\mathcal{A}}$ and define $\bm\eta^{(\ell),j}=\bm\xi^{(\ell),j}-\bm\theta $. By definition, we have that
$$ 
 \mathbb{E}\left[ \ell(S-\bm \xi^{(\ell),j} \cdot \mathbf{Y})\right]= \mathbb{E}\left[ \ell(S-\bm \theta \cdot \mathbf{Y}-\bm \eta^{(\ell),j }\cdot \mathbf{Y})\right].
$$
Therefore, $\bm \eta^{(\ell),j}$ should be a convex hedging strategy for the residual risk, otherwise this would contradict the optimality of $\bm \xi^{(\ell),j}_S $. Analogously, if $\bm\eta^{(\ell),j}$ is a convex hedging strategy for the residual risk then $\bm \xi^{(\ell),j}_S=\bm\theta+\bm\eta^{(\ell),j}$ must be a convex hedging strategy for $S$.
\end{enumerate}
\end{proof}}
\begin{remark}\textcolor{black}{Note that, if the loss function $\ell$ is strictly convex, then the minimiser in \eqref{eta} is unique and therefore only part $a)$ of Lemma \ref{hedgingdecomposition} applies.}
\end{remark}

\textcolor{black}{The two-step valuation  equals the cost of the quadratic hedging strategy $\boldsymbol{\theta}$ plus the cost of capital of the hedging strategy $\bm \eta^{(\ell)}$ necessary to cover the residual risk. From an economic standpoint, the insurer uses the hedging strategy $\bm\xi^{(\ell)}$ to cover the regulatory requirements and shares its cost between the policyholders and shareholders. From an
asset-liability perspective, the interpretation is as follows: let $\bm \xi^{(\ell)}_S\cdot \bm y= (\bm\theta+\bm\eta^{(\ell)})\cdot\bm y$ be the cost
of the convex hedging strategy at time 0 where $\bm\theta\cdot \bm y$ is the cost of the quadratic hedging at
time 0. The two-step insurance valuation assumes that the quadratic hedging cost $\bm\theta\cdot \bm y$ is
borne by policyholders and the residual cost $\bm\eta^{\ell} \cdot\bm y$ to achieve a VaR-neutral portfolio is borne by shareholders. {Shareholders require a return for the  capital they provide (the so called cost-of-capital) $i\cdot\bm\eta^{(\ell)}\cdot \bm y$ which is also paid by the policyholders}. Thus, what we define as the fair value (the quadratic hedging cost plus the cost-of-capital risk margin) is $\rho(S)=\bm \theta \cdot \bm y+i\cdot\bm\eta^{\ell}\cdot \bm y$. } \textcolor{black}{Note that we assume that assets are invested according to the hedging strategy $\xi^{(\ell)}_S$, which covers the full liability. The split into the strategies $\bm \theta $ and $\eta^{\ell} $ is notional, with the specific purpose of apportioning the cost of hedging to different stakeholders.} \textcolor{black}{Our valuation is thus closely related to the two-step approach of \cite{mohr2011market} who also considers the apportionment of hedging costs to policyholders and shareholders in a very similar way. \cite{mohr2011market} employs a different acceptability condition, interpreting the
cost-of-capital rate as the expected excess return required by the shareholders, while, in our case, we implicitly assume that shareholders bear the residual risk of deviation from the mean-quantile hedging strategy.}

The remainder of the paper focuses on the two-step valuation where the quantile hedging strategy is considered in the second step, called the \textit{mean-quantile valuation}. From now on, we drop the upper-script $\ell$ if we consider quantile hedging in the second step. We also briefly discuss the two-step valuation with the expectile loss function \eqref{expectileloss} in the second step, which we call the \textit{mean-expectile valuation}.

In the cost-of-capital approach for \eqref{rhococ}, a capital $c=\VaR_\alpha (R(S,\bm \theta))$ is set up and kept risk-free until year 1 to guarantee that
$$
\VaR_\alpha (S-\bm \theta \cdot \mathbf{Y}-c)=0.
$$
In the two-step valuation proposed in this section, we set up a strategy $\bm \eta$ such that 
$$\VaR_\alpha (S-\bm \theta \cdot \mathbf{Y}-\bm \eta \cdot \mathbf{Y})=\VaR_\alpha\Big(R\big(R(S,\bm \theta),\bm \eta\big)\Big)=0. $$


In the following result, we show that the mean-quantile valuation is positive homogeneous and fair. \begin{theorem}\label{theoremfair}
The mean-quantile valuation is positive homogeneous and  fair.
\end{theorem}
\begin{proof}The positive homogeneity is directly obtained by the positive homogeneity of $\bm \theta$ and $\bm \eta$  \citep[see][]{Koenker2005}.
To prove that the mean-quantile valuation is fair, we show that the valuation is market-consistent and actuarial.
\begin{itemize}
\item
First, we notice that the solution of the quadratic hedging problem is additive:
\begin{equation*}
\bm \theta_{S+S^h}= \bm \theta_{S}+\bm \theta_{S^h}.
\end{equation*}
Since $S^h$ can be hedged with $\bm \nu$, we have that $\bm \theta_{S^h}=\bm \nu$. Therefore, we find that
\begin{align*}
\rho (S+S^h)&=\bm \nu\cdot \mathbf{y}+ \bm \theta_{S} \cdot \mathbf{y}+ i \ \bm \eta_{S+S^h} \cdot \mathbf{y}\\
&=\bm \nu\cdot \mathbf{y}+\rho (S)
\end{align*}
where $\eta_{S+S^h}$ is the quantile hedging strategy of the residual loss of $S+S^h$:
\begin{equation*}
R(S+S^h, \bm \theta_{S+S^h})=S+S^h-\bm \theta_{S^h} \cdot \mathbf{Y}-\bm \theta_{S} \cdot \mathbf{Y}=S-\bm \theta_{S} \cdot \mathbf{Y}=R(S,\bm \theta_S),
\end{equation*}
which ends the proof for the market-consistency.

\item 
 By standard least-squares arguments, the quadratic hedging strategy of $S^\perp$ is
\begin{equation*}
	\bm \theta_{S^\perp}=(\mathbb{E}[S^\perp],0,\dots,0).
\end{equation*}
Otherwise stated, if a liability is independent of risky assets, the hedging strategy only invests risk-free. Therefore, we find that 
\begin{equation*}
\rho (S^\perp)=e^{-r} \mathbb{E}[S^\perp]+ i \ \bm \eta_{S^\perp} \cdot \mathbf{y},
\end{equation*}
where $ \eta_{S^\perp} $ is the quantile hedging strategy for $R(S^\perp, \bm \theta_{S^\perp})=S^\perp-\mathbb{E}[S^\perp]$. Since $R(S^\perp, \bm \theta_{S^\perp})$ is independent of the risky assets, we find that the quantile hedging strategy is (cf. Theorem 4 in \citealp{dhaene2017fair})
\begin{equation*}
\bm \eta_{S^\perp}=(\VaR_\alpha(S^\perp-\mathbb{E}[S^\perp]),0,\dots,0),
\end{equation*}
which implies that the two-step valuation of $S^\perp$ is
\begin{equation*}
\rho (S^\perp)=e^{-r} \mathbb{E}[S^\perp]+e^{-r} i \VaR_\alpha(S^\perp-\mathbb{E}[S^\perp]).
\end{equation*}
The mean-quantile valuation corresponds then to the standard cost-of-capital principle and the mean-quantile valuation is fair.
\end{itemize}
\end{proof}
\begin{theorem}\label{theoremfairexpectiles}
The mean-expectile valuation is positive homogeneous 
 and fair.
\end{theorem}
\begin{proof}To verify the positive homogeneity of the valuation it is sufficient to check that  the expectile strategy is positive homogeneous. Note that for any $a>0,~S\in\mathcal{C}, \bm\beta\in\mathcal{B}$ we have:
$$
\E[\ell_\tau(aS-a\bm\beta\cdot\bm Y)]=a^2\E[\ell_\tau(S-\bm\beta\cdot\bm Y)],
$$
that implies $\arg\min_{\bm\beta\in\mathcal{B}}\E[\ell_\tau(aS-\bm\beta\cdot\bm Y)]=a\arg\min_{\bm \beta\in\mathcal{B}}\E[\ell_\tau(S-\bm\beta\cdot\bm Y)]$. The fairness of the mean-expectile valuation follows  exactly the same steps of the one for the mean-quantile valuation and is therefore omitted.
\end{proof}

Hereafter, we show that applying quantile hedging to the residual risk will reduce the tail of the residual risk compared to making an investment in the risk-free asset. \newline
Let us assume that we want to hedge $R(S,\bm \theta)=S-\bm \theta \cdot \mathbf{Y}$ and the regulator imposes that $\VaR_\alpha(R(S,\bm \theta)-\bm \beta \cdot \mathbf{Y}) =0$ for some trading strategy $\bm\beta\in\mathcal{B}$. To achieve this, there are two possibilities:
\begin{itemize}
\item Consider an investment in the risk-free asset equal to $\VaR_\alpha(R(S,\bm \theta )) $. We denote this strategy by $\bm \nu $.
\item Consider the quantile hedging strategy $\bm \eta $ such that $\VaR_\alpha(R(S,\bm \theta )-\bm  \eta \cdot \mathbf{Y})=0$.
\end{itemize}
The quantile hedging strategy is the minimiser of the Tail Value-at-Risk (TVaR) deviation of the residual risk among all the strategies which satisfy the VaR regulatory constraint. We recall that TVaR is a coherent risk measure defined as 
$$
\textrm{TVaR}_{\alpha}(X)=\frac{1}{1-\alpha}\int_\alpha^1\VaR_u(X)\textrm{d}u,\quad\textrm{for any finite mean risk }X~\textrm{and any }\alpha\in(0,1)
$$
and the  TVaR deviation (dTVaR) is defined 
as $\dTVaR_\alpha(X)=\TVaR_\alpha(X)-\E(X)$ for any $X\in\mathcal{C}$.
\begin{lemma}\label{lemmatvar}
Consider $\bm \theta $, $\bm \nu $ and $\bm \eta $ as defined above. The quantile hedging strategy is the minimiser of the TVaR deviation of the residual risk: 
\begin{equation*}
\dTVaR_\alpha(R(S-\bm  \theta\cdot \mathbf{Y},\bm  \eta)) \leq \dTVaR_\alpha(R(S-\bm  \theta\cdot \mathbf{Y},\bm  \beta)),
\end{equation*}
for all hedging strategies $\bm  \beta $ such that $\VaR_\alpha(R(S,\bm \theta )-\bm \beta \cdot \mathbf{Y}) =0$. This is in particular the case for $\bm  \beta= \bm \nu$.
\end{lemma}
\begin{proof}
This is a direct application of Theorem 3.2. in \cite{rockafellar2008risk}.
\end{proof}

\textcolor{black}{Lemma \ref{lemmatvar} also reveals why quantile hedging, through the minimisation of a convex loss function, does avoid perverse incentives, as shown in Example \ref{example: regulatory arbitrage}. In the next example, we illustrate the application of our two-step valuation approach in an insurance portfolio.}

\begin{example}\label{exampleELoneperiod}
Here we show how our proposed valuation methodology works in a single-period setting, where the liability $S$ is highly -- but  non-linearly -- correlated with a tradeable asset $Y_{1}$. A more realistic dynamic version of this model is discussed in detail in Section \ref{numericalexamples}.

We consider a portfolio of equity-linked life insurance contracts, which guarantee a survival benefit to all policyholders who are still alive at maturity time $T=1$. The insurance liability can be expressed as
\begin{equation}
S=N\times\max\left(  Y_1,K\right), \label{payoff_single_period}
\end{equation}
where $N$ represents the number of survivors at time $1$,  
among an initial population of $1000$ policyholders and $K$ is a guarantee level. We make the following assumptions: $N\sim Bin(1000,0.9)$; $Y_1 \sim LN(0.1,0.2^2)$; the value of the risky asset at time 0 is $y_1=1$; interest rates are zero, that is, $y_0=Y_0=1$ and $K=1$. The analysis is carried out on a sample of 200,000 simulated scenarios.

On the left of Figure \ref{fig:Example2_a}, we plot samples of $S$ against $Y_1$. One can see the strong (nearly deterministic) positive relationship between the two, indicating that $Y_1$ can be used to hedge $S$. On the right of Figure \ref{fig:Example2_a}, we show the residuals of the quadratic hedging strategy, $R(S,\bm\theta)$ against $Y_1$. The residuals are clearly not independent of the tradeable asset, indicating that quantile or expectile hedging of those residuals will be meaningful. In other words, we would expect that our two-step valuation $\rho$ would give different answers than the valuation $\phi$ defined in \eqref{rhococ}.

\begin{figure}
	\centering
\begin{knitrout}
\definecolor{shadecolor}{rgb}{0.969, 0.969, 0.969}\color{fgcolor}
\includegraphics[width=1\linewidth]{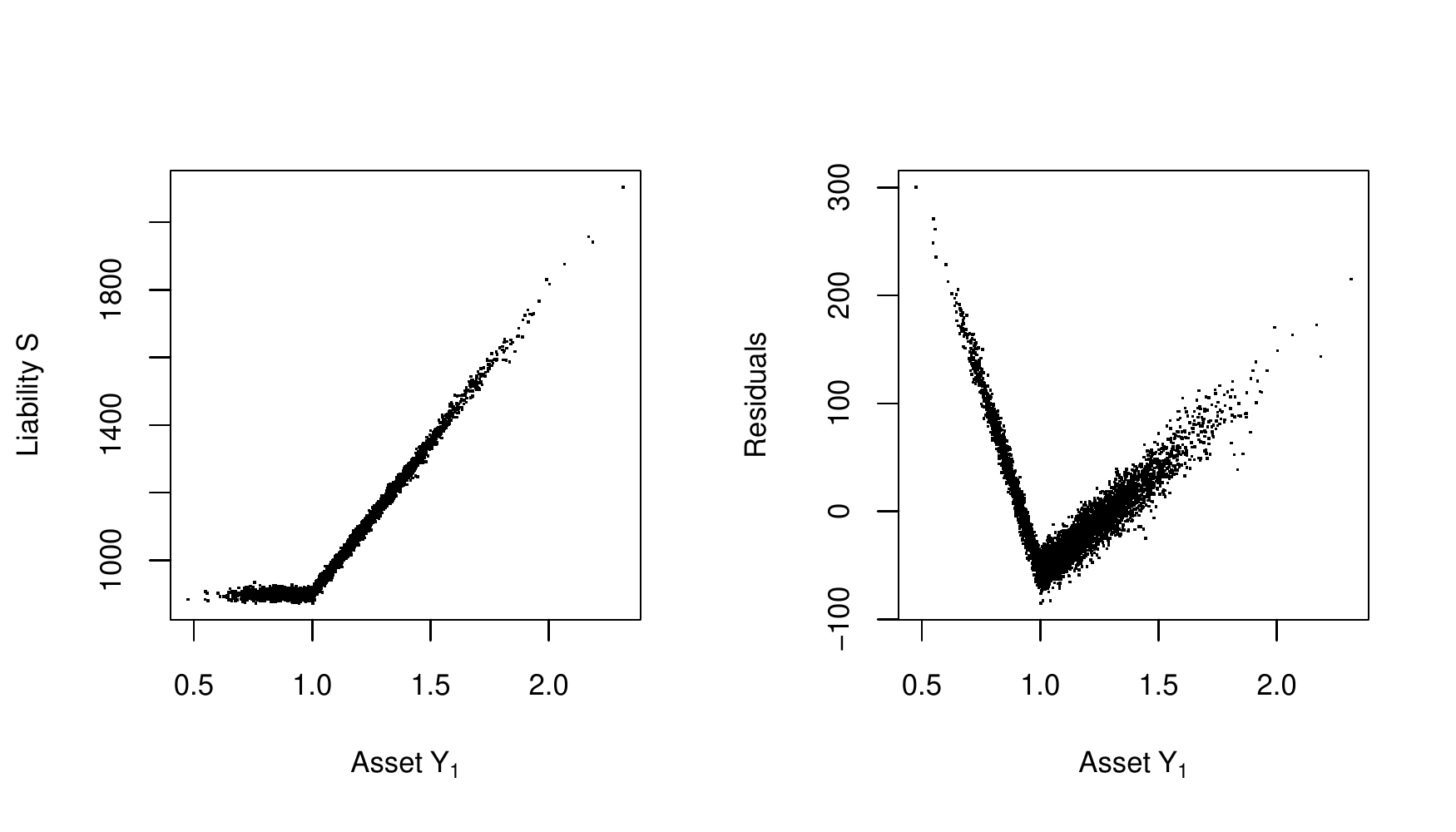} 

\end{knitrout}
	\caption{Liability $S$ (left) and residual $R(S,\bm\theta)$ (right) against value of the risky asset $Y_1$.}
	\label{fig:Example2_a}
\end{figure}

In Table \ref{tab:Example2_a} we present the following trading strategies:
\begin{enumerate}
 \item $\bm \theta$: the result of quadratic hedging of $S$.
 \item $\bm \xi$: the result of quantile hedging of $S$, with $\alpha=0.99$.
 \item $\bm \xi^{(\ell_\tau)}$: the result of expectile hedging of $S$, with $\tau=0.998$. The value of $\tau$ was calibrated so that the Value-at-Risk and expectiles of $S$ match, i.e. $\VaR_\alpha(S)= \Gamma^{(\ell_\tau)}(S)$.
 \item $(\VaR_\alpha(R(S,\bm \theta)),0)$: investing the VaR of the residuals $R(S,\bm \theta)$ in the risk-free asset.
 \item $\bm \eta$: the result of quantile hedging of $R(S,\bm \theta)$.
 \item $\bm \eta^{(\ell_\tau)}$: the result of expectile hedging of $R(S,\bm \theta)$.
\end{enumerate}
We can observe that $\bm \xi$ and $\bm \xi^{(\ell_\tau)}$ place a substantially higher investment into the risk-free asset, compared to $\bm \theta$, reflecting the more stringent criterion encoded in the respective loss functions -- recall that $\E\big[R(\bm \theta,S)\big]=0$, $\VaR_\alpha\big(R(\bm \xi,S)\big)=0$, $\Gamma^{(\ell_\tau)}\big(R(\bm \xi^{(\ell_\tau)},S)\big)=0$. At the same time, the investment in the risky asset is somewhat higher for $\bm \theta$, reflecting a lower sensitivity to adverse movements in asset values. The remaining three strategies pertain to hedging the residuals of the first (quadratic hedging) step. We see that this second step, for quantile and expectile hedging, involves a reduction in the exposure to the risky asset. \textcolor{black}{From Table \ref{tab:Example2_a}, we also observe that the quantile hedging strategy can lead to a reduction of the total assets required at time 0 to achieve a VaR-neutral portfolio. Indeed, the cost of quantile strategy $\boldsymbol{\xi}$ is lower than the cost of the quadratic strategy and investing the residual risk risk-free.}
\begin{table}[h]
  
    \centering  
        \caption{Investment in risk-free and risky asset, from hedging strategies associated with the two-step valuation of $S$.}\label{tab:Example2_a}
    \begin{tabular}{cccc}
    \hline
         Strategy & risk-free asset  & risky asset & cost of strategy \\\hline
          $\bm \theta$  & 247 & 709 & 956 \\
          $\bm \xi$ &   460 & 658  &   1118 \\
          $\bm \xi^{(\ell_\tau)}$ &  450 & 663 &  1113 \\\hline
          $(\VaR_\alpha(R(S,\bm \theta)),0)$ & 163 & 0 & 163\\
          $\bm \eta$ & 213 & -52& 161\\
          $\bm \eta^{(\ell_\tau)}$ &204 &  -47 &157 \\\hline
    \end{tabular}
    
\end{table}

In Figure \ref{fig:Example2_b}, we show the densities of the residuals corresponding to the strategies of 4.--6. above. We can see that quantile and expectile hedging lead to a somewhat different shape, compared to a quadratic regression that is followed by investing the $\VaR_\alpha$ of the residual risk in the risk-free asset. We quantify the difference between the three densities, by stating the corresponding TVaR deviations:

\begin{align*}
 \dTVaR_\alpha\big(R(S,\bm \theta)-\VaR_\alpha(R(S,\bm \theta))\big)&= 194.2\\
 \dTVaR_\alpha\big(R(R(S,\bm \theta),\bm \eta)\big)&= 182.5\\
 \dTVaR_\alpha\big(R(R(S,\bm \theta),\bm \eta^{(\ell_\tau)})\big)&= 182.6
\end{align*}
Hence, the application of quantile and expectile hedging reduces the tail of the residuals, compared to the quadratic hedging case.  For quantile regression, this observation is a direct implication of Lemma \ref{lemmatvar}. Of course different conclusions may be reached, if the criterion for measuring the variability of residuals changes. For example, considering the standard deviation of residuals privileges quadratic regression, leading to $\sigma\big(R(S,\bm \theta) - \VaR_\alpha(R(S,\bm \theta)) \big)=49.0$, while $\sigma\big(R(R(S,\bm \theta),\bm \eta ) \big)=50.3$; note though the difference between the two is small.

\begin{figure}
	\centering

\begin{knitrout}
\definecolor{shadecolor}{rgb}{0.969, 0.969, 0.969}\color{fgcolor}
\includegraphics[width=0.8\linewidth]{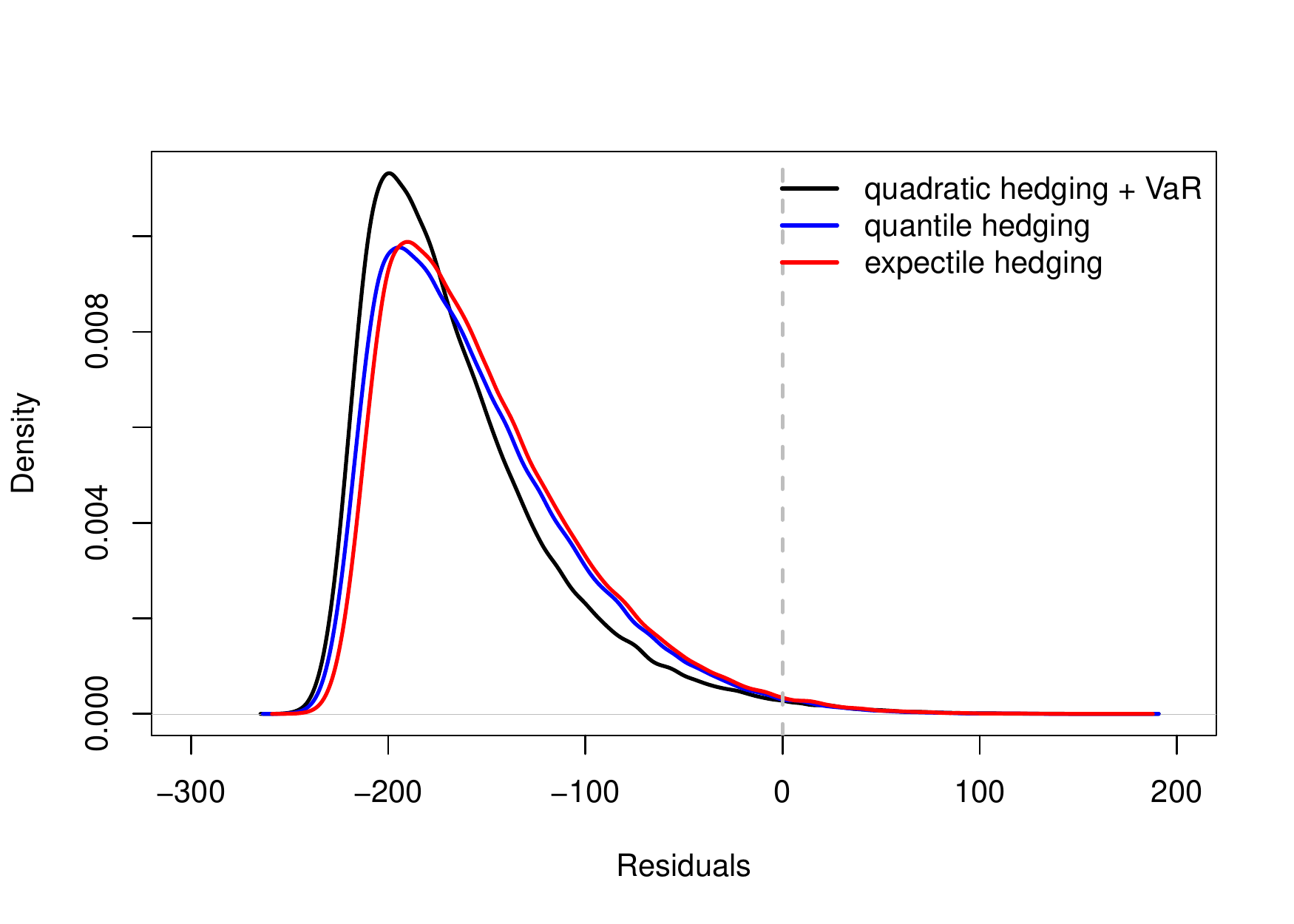} 

\end{knitrout}

	\caption{Densities of residuals, for quadratic hedging of $S$, followed by investing the VaR of the residual in the risk-free asset (blue); quantile hedging of $S$ (blue); and expectile hedging of $S$ (red).}
	\label{fig:Example2_b}
\end{figure}

Finally, we state the fair value of $S$, as calculated via the three different hedging approaches 4.-6., for cost-of-capital rate $i=0.1$:
\begin{align*}
\phi(S)&=\bm \theta\cdot \mathbf{y}+i \cdot \VaR_\alpha\big( R(S,\bm \theta)\big) =  972.6 \\
\rho(S)&=\bm \theta\cdot \mathbf{y}+i  \cdot \bm \eta\cdot \mathbf{y} =  972.4 \\
\rho^{(\ell_\tau)}(S)&=\bm \theta\cdot \mathbf{y}+i \cdot  \bm \eta^{(\ell_\tau)}\cdot \mathbf{y} =  972. 
\end{align*}
Hence, in this example, the impact on the valuation of $S$ is very limited, even though the hedging strategies and, importantly, the statistical behaviour of residuals, are different.

\end{example}

\section{Two-step valuation in a  multi-period setting} \label{sectionmultiperiod}

We extend the previous setup to a model over multiple
time periods. We consider a filtered probability space: 
$(\Omega,(\mathcal{F}_t)_{t\in\{0,1,\ldots,T\}},\p)$, where $\mathcal{F}_0=\{\emptyset,\Omega\}$, $\mathcal{F}_T=\mathcal{F}$ and  the $\sigma$-algebra $\mathcal{F}_t$ represents the information available up to time $t$, $t\in\{0,1,\ldots,T\}$. 

Again we consider  $n+1$ traded assets. We denote    $\mathbf{Y}(t)=(e^{rt},Y_1(t),\dots,Y_n(t))$ the vector of asset prices at time $t\in\{1,2,\ldots,T\}$ and assume that the asset portfolio can be freely reset at each time $t$, hence we do not require  trading strategies to be self-financing. A trading strategy is an $n+1$ vector $\bm{\beta}(t)=(\beta_0(t),\ldots,\beta_n(t))$, where each $\beta_i(t)$ is predictable (i.e. $\mathcal{F}_{t-1}$-measurable) and represents the funds invested in asset $i$ during the time interval $[t-1,t)$, for any $i=0,1,\ldots,n$ and $t\in\{1,2,\ldots,T\}$. We denote $\mathcal{B}(t)$ the set of all real-valued $\mathcal{F}_{t-1}$-measurable trading strategies available to the investor for the time interval $[t-1,t)$.

In this section, we study the problem of an insurer who needs to determine the fair valuation at any time $t<T$ for an insurance liability $S$ which matures at time $T$. In \cite{barigou2019fair}, this objective was achieved by a backward iteration in which for each time period the liability is hedged by quadratic hedging and the non-hedged residual part is priced via an actuarial valuation, e.g. the standard deviation principle. While this approach is fair in the sense of Theorem \ref{theoremfair}, it does not take into account the yearly solvency requirement in its valuation, namely that the hedging portfolio should cover the fair value of the liability with a confidence level $\alpha$. 

Here, we extend our two-step hedging approach from a one-period to a multi-period setting. In the first step of the valuation, a quadratic hedging strategy is set up for the fair value of the liability in the next period and we know by the relation \eqref{expproperty} that its expected payoff corresponds to the expected liability. In the second step of the valuation, we apply a quantile hedging strategy to the residual loss, and by construction, the yearly solvency capital requirement will be satisfied. The cost of this strategy is then included in the valuation through an appropriate cost-of-capital risk margin. 

\textcolor{black}{In our local risk-minimisation framework, hedging is carried out afresh at the beginning  of each time period. As a result, at intermediate times excess capital may be released or the need may arise for capital to be injected by the company’s shareholders. Nonetheless,} \textcolor{black}{as will be discussed below, such rebalancing costs would be typically met from the invested shareholder capital of the previous period.}

\subsection{Fair valuation by iterated  two-step valuation}

Consider an insurance liability $S$ which matures at time $T$. The quadratic and quantile hedging strategies at time $T-1$ are determined by 
\begin{align*}
&\bm \theta(T)=\arg \min _{\bm\beta \in \mathcal{B}(T)} \mathbb{E}_{T-1}\left[(S-\boldsymbol{\beta}(T) \cdot \mathbf{Y}(T))^2\right],\\
&\bm \eta(T)=\arg \min _{\bm \beta \in \mathcal{B}(T)} \mathbb{E}_{T-1}\left[ \ell_{\alpha}(S-\boldsymbol{\theta}(T) \cdot \mathbf{Y}(T)-\boldsymbol{\beta}(T) \cdot \mathbf{Y}(T)) \right],
\end{align*}
where $ \ell_{\alpha}$ is the Koenker-Bassett error given in \eqref{penfunction}. By the properties of quantile hedging, the payoff of both hedging strategies will cover the liability with a confidence level of $\alpha$, hence satisfying the regulatory constraint: 
\begin{equation*}
	\text{VaR}_{\alpha,T-1}\left(S-\boldsymbol{\theta}(T) \cdot \mathbf{Y}(T)-\boldsymbol{\eta}(T) \cdot \mathbf{Y}(T)\right)=0.
\end{equation*}
The fair value at time $T-1$ of the liability is then defined as the cost of the quadratic hedging strategy and a cost-of-capital risk margin for the quantile hedging strategy: 
\begin{equation*}
\rho_{T-1}(S)=\boldsymbol{\theta}(T) \cdot \mathbf{Y}(T-1)+i\ \boldsymbol{\eta}(T) \cdot \mathbf{Y}(T-1).
\end{equation*}
As in the static setting, the cost of the quadratic hedging strategy is covered by policyholders, while the cost of the quantile hedging strategy (interpreted as a capital requirement) is provided by shareholders, who require an interest $i$ for their investment. We can now repeat iteratively the two-step valuation until we reach the fair value at time 0, at each step hedging the fair value of one period ahead. For the fair value $\rho_{t+1}(S)$, both hedging strategies are given by
\begin{align*}
\bm \theta(t+1)&=\arg \min _{\bm\beta \in \mathcal{B}(t)} \mathbb{E}_t\left[(\rho_{t+1}(S)-\boldsymbol{\beta}(t+1) \cdot \mathbf{Y}(t+1))^2\right]\\
\bm \eta(t+1)&=\arg \min _{\bm \beta \in \mathcal{B}(t)} \mathbb{E}_t\left[\ell_{\alpha}(\rho_{t+1}(S)-\boldsymbol{\theta}(t+1) \cdot \mathbf{Y}(t+1)-\boldsymbol{\beta}(t+1) \cdot \mathbf{Y}(t+1))\right].
\end{align*}
Then, the fair value at time $t$ is given by 
\begin{equation}
\rho_{t}(S):=\boldsymbol{\theta}(t+1) \cdot \mathbf{Y}(t)+i\ \boldsymbol{\eta}(t+1) \cdot \mathbf{Y}(t),\label{fairvaluet}
\end{equation}
and the yearly solvency constraints are satisfied by construction: 
\begin{equation}\label{varconstraint}
	\text{VaR}_{\alpha,t}\left(\rho_{t+1}(S)-\boldsymbol{\theta}(t+1) \cdot \mathbf{Y}(t+1)-\boldsymbol{\eta}(t+1) \cdot \mathbf{Y}(t+1) \right)=0,\quad \forall t\in \{0,1,\dots,T-1\}.
\end{equation}

\textcolor{black}{We consider now the issue of the potential need for capital  injections at intermediate times. Consider for example the assets at the end of the time period $t$,  $\bm \theta(t)\cdot \textbf{Y}(t) +\bm\eta(t)\cdot \textbf{Y}(t)=\bm\xi(t)\textbf{Y}(t)$. Then, shareholders need to inject more capital if $\bm\xi(t)\cdot \textbf{Y}(t)<\bm\xi(t+1)\cdot \textbf{Y}(t)$ or, conversely, capital can be released if  $\bm\xi(t)\cdot \textbf{Y}(t)>\bm\xi(t+1)\cdot \textbf{Y}(t)$. Let $RB(t) = \bm\xi(t+1)\cdot \textbf{Y}(t)-\bm \xi(t)\cdot \textbf{Y}(t)$ stand for the rebalancing cost at time $t$. Focusing on a quantile hedging strategy with a high confidence level $\alpha$, VaR-neutrality implies that 
	\begin{equation}\label{eq:rebal}
		\mathbb P(\bm\theta(t)\cdot \textbf{Y}(t)+\bm \eta(t)\cdot \textbf{Y}(t)  -\rho_t(S)\geq  0)= \alpha \Leftrightarrow \mathbb P(\bm \eta(t+1)\cdot \textbf{Y}(t)(1-i) \geq RB(t) )= \alpha
	\end{equation}	Hence, with high probability, the funds $RB(t)$ required to keep satisfying the regulatory requirement at time $t$ is lower than $\eta(t+1)\cdot \textbf{Y}(t)(1-i)$, which represents the shareholder capital needed to recapitalise the portfolio from scratch, allowing also for the cost of that capital. In the unlikely event that $\bm \eta(t+1)\cdot \textbf{Y}(t)(1-i) < RB(t)$, necessary capital cannot be raised and the procedure is stopped. }

Time-consistency is an important concept for characterising the relationship between  different static valuations. It means that the same value is assigned to a liability regardless of whether it is calculated in one step or in two steps backwards in time. 
\begin{definition}
A sequence of valuations $\left(\rho_{t}\right)_{t=0}^{T-1}=\{\rho_0,\dots,\rho_{T-1}\}$ is time-consistent if:
\begin{equation}\label{tcrequirement}
\rho_{t}(S)=\rho_{t}\left(\rho_{t+1}(S)\right), \quad \text { for any liability } S \text { and } t=0,1, \ldots, T-2. 
\end{equation}
\end{definition}

By construction, our valuation framework is time-consistent. Indeed, from  \eqref{fairvaluet} we see that the fair valuation at time $t$ is obtained by applying the one-period two-step valuation on $\rho_{t+1}$, which itself comes from the two-step valuation on $\rho_{t+2}$ and so on. Therefore, the time-consistency condition \eqref{tcrequirement} is directly satisfied.

We now provide a simple example of our multi-period two-step valuation in a multivariate normal setting. In such a framework, explicit solutions for the fair valuation can be obtained. 

\begin{example}\label{normalexample}

We consider a multi-period model with two assets only: $\bY(t)=(Y_{0}(t),Y_{1}(t))$,  $t=0,1,\dots, T$. We assume that $Y_{0}(t)=1$ for all $t$ (i.e. the risk-free asset has zero interest rate) and  that the returns of the risky asset, $R_t=\frac{Y_{1}(t)}{Y_{1}(t-1)},~t=1, \dots, T$ are i.i.d. Moreover, we write the liability development of $S$ as
\begin{align*}
    S=s_0+S_1+\dots+S_T,
\end{align*}
where $s_0$ is constant, $S_t$ is the liability development at time $t$ which is $\cF_t$-measurable, $S_t$'s are independent, and $\E[S_t]=0$ for $t=1, \dots, T, $ we remark that such a decomposition was also considered in \cite{tsanakas2013market}. Furthermore, $(R_t,S_t)$ are bivariate normally distributed, with constant correlation $c$. For the sake of brevity we define:
\begin{align*}
    \kappa ~&= \frac{\E[R_t-1]}{\sigma(R_t)},\\
    \gamma_t &= \sigma(S_t),
\end{align*}
where $\kappa$ is constant, while $\gamma_t$ is deterministic but time varying. First, we determine the quadratic hedging strategy at time $T-1$:
\begin{equation*}
	\bm \theta(T):=\argmin_{\bm \beta\in\mathcal{B}(T)}\E_{T-1}\left[(S-\bm \beta\cdot \bY(T))^2\right].\\
\end{equation*}
By standard conditional least-squares arguments, we get:
    \begin{align*}
    \theta_{1}(T)&=\frac{\text{Cov}_{T-1}(S,Y_{1}(T))}{\text{Var}_{T-1}(Y_1({T}))},\\
    \theta_{0}(T)&=\E_{T-1}[S]-\theta_{1}(T)\E_{T-1}[Y_{1}(T)],
\end{align*}
\textcolor{black}{with the property: 
\begin{equation*}
	\E_{T-1}[S]=\E_{T-1}[\bm \theta(T) \cdot \bY(T)],
\end{equation*}}
so that the cost of the quadratic hedging strategy is given by:\begin{align*}
    \bm \theta(T) \cdot \bY(T-1)&=\E_{T-1}[S]-\frac{\text{Cov}_{T-1}(S,Y_{1}(T))}{\text{Var}_{T-1}(Y_{1}(T))}\left(\E_{T-1}[Y_{1}(T)]-Y_{1}(T-1)\right)\\
    &=s_0+\sum_{j=1}^{T-1}S_j-\kappa \gamma_T c.
\end{align*}
In a second step, we determine the risk margin by computing the quantile hedging strategy for the residual:
\begin{align*}
    \bm\eta(T)&:=\argmin_{\bm \beta\in\mathcal{B}(T)} \E_{T-1}\left[ \ell_{\alpha}(S-\bm\theta(T)\cdot \bY(T)-\bm\beta\cdot \bY(T)) \right].
\end{align*}
By the normality assumption, we have that 
$$
(S-\bm\theta(T)\cdot\bY(T))\perp Y_{1}(T)\implies \eta_{1}(T)=0,
$$
i.e. there is no investment in the risky asset. The related cost is then given by
\begin{align*}
 \bm \eta(T)\cdot \bY(T-1)&=\VaR_{\alpha,T-1} (S-\bm\theta(T)\cdot \bY(T)),\\
&=\E_{T-1}[S-\bm\theta(T)\cdot \bY(T)]    +\lambda \sigma_{T-1}[S-\bm\theta(T)\cdot \bY(T)],\qquad(\lambda:=\Phi^{-1}(\alpha))\\
&=\lambda \gamma_T\sqrt{1-c^2}.
    \end{align*}
The resulting fair value of $S$ is then given by
\begin{align*}
    \rho_{T-1}(S)&=\bm \theta(T)\cdot \bY(T-1)+i\bm \eta(T)\cdot \bY(T-1)\\
    &=s_0+\sum_{j=1}^{T-1}S_j-\kappa \gamma_T c+i\lambda \gamma_T\sqrt{1-c^2}.
\end{align*}
By noting that with respect to $\cF_{T-2}$, the only random element in $\rho_{T-1}(S)$ is $S_{T-1}$, we find that 
 \begin{align*}
	\rho_{T-2}(S)&=  \bm \theta(T-1)\cdot \bY(T-2)+i\bm \eta(T-1)\cdot \bY(T-2)  \\
	&=s_0+\sum_{j=1}^{T-2}S_j-\kappa (\gamma_T+\gamma_{T-1}) c+i\lambda (\gamma_T+\gamma_{T-1})\sqrt{1-c^2}.
\end{align*}

The cost of capital of \textbf{all} future capital requirements is used when valuing the liability at a particular time, which is in agreement with the Solvency II risk margin. An inductive argument would lead to:
  \begin{equation*}
   \rho_{0}(S)= s_0 -\kappa c \sum_{j=1}^T \gamma_j +i\lambda \sqrt{1-c^2}\sum_{j=1}^T \gamma_j.\\
 \end{equation*}

Hence, the fair value of the liability is composed of three terms: first the expected liability $s_0$, a second term accounting for the dependence between the excess risky asset returns and the liability increments, and, third, is the cost-of-capital risk margin which takes into account the non-hedgeable risk. We note that the risk margin vanishes as $c\to1$ since the liability can be completely hedged in this case.
 \end{example}

\section{Dynamic hedging by neural networks}\label{sectionneuralnetwork}

The backward recursive scheme presented above is similar to the one solving the local quadratic hedging problem (\citealp{follmer1988hedging}), which is usually implemented by dynamic programming. Since the optimal hedging strategy is a function of conditional expectations, a popular technique consists of constructing a Markov grid with the use of a multinomial tree model for the risky asset dynamics (see e.g. \citealp{augustyniak2017assessing}, \citealp{coleman2006hedging} and \citealp{godin2016minimizing}). The Markov property is key to reducing the dimensionality of the dynamic programming algorithms, because it implies that conditional expectations with respect to $\mathcal{F}_{t}$ reduce to conditional expectations with respect to assets prices at time $t$ only, i.e. $\bY(t)$.

In this paper, we present a general procedure to solve the dynamic quadratic-quantile hedging problem in a Markovian setting. The procedure involves an iterated non-linear optimisation which is solved by neural networks. We note however that other non-linear regression methods (such as gradient boosted trees) can be used as well. \textcolor{black}{As suggested by one referee, another possibility is to solve the algorithm via a Least-Square Monte-Carlo (LSMC) approach similar to \cite{ahmad2020}}. Moreover, we present the algorithm for the mean-quantile valuation but, in principle, the whole procedure remains valid for any loss function $\ell$.

\subsection{General algorithm for the dynamic hedging problem}

We recall that the iterated dynamic hedging problem is given by 
\begin{align}\label{iteratedsystem}
\begin{split}
\bm \theta(t+1)&:=\argmin_{\bm \beta\in\mathcal{B}(t+1) }\E_{t}\left[(\rho_{t+1}(S)-\bm \beta \cdot \bY(t+1))^2\right]\\
\bm\eta(t+1)&:=\argmin_{\bm \beta\in \mathcal{B}(t+1)} \E_{t}\left[ \ell_{\alpha}(\rho_{t+1}(S)-\bm\theta(t+1)\cdot \bY(t+1)-\bm\beta\cdot \bY(t+1)) \right]\\
\rho_{t}(S)&:=\boldsymbol{\theta}(t+1) \cdot \mathbf{Y}(t)+i\ \boldsymbol{\eta}(t+1) \cdot \mathbf{Y}(t)
\end{split}
\end{align}
for any $t=T-1,\dots,0$, starting with $\rho_{T}(S)=S$ and $\ell_{\alpha}$ is the Koenker-Bassett error \eqref{penfunction}.

From now on, we assume that there exists an $m$-dimensional process $ \bm Z(t)$ which drives all the processes of interest. In an insurance context, $Z(t)$ may represent for instance the asset processes $Y(t)$ and the number of policyholders alive at time $t$. The filtration $\mathcal{F}$ is generated by all observations about the process $\bm Z$: $\mathcal{F}_t=\sigma(\bm Z(u)\mid u \leq t)$.

To avoid path-dependence issues in the hedging framework and reduce the complexity of the dynamic hedging algorithm, we make the standard assumption that $\bm Z$ is Markovian. We note that many standard financial and actuarial processes do follow the Markov property.
\begin{assumption}\label{markovassumption}
	$\bm Z$ has the Markov property with respect to the filtration $\mathcal{F}$, i.e. $$\mathbb{P}(\bm Z(t+1)\leq \bm x\mid \mathcal{F}_t)=\mathbb{P}(\bm  Z(t+1)\leq \bm x\mid \bm Z(t)).$$
\end{assumption}
By  Assumption \ref{markovassumption}, the candidate hedging strategies in the dynamic hedging problem \eqref{iteratedsystem} can be expressed as $g(\bm Z(t))$ where $g: \mathbb{R}^{m}\rightarrow \mathbb{R}^{n+1}$ is a function which takes the random process at time $t$ as inputs and outputs the hedging positions in the $(n+1)$ assets. \textcolor{black}{Indeed, from the Markov property, the optimal strategy $\boldsymbol{\theta}(t+1)$ and $\boldsymbol{\eta}(t+1)$ for the period $[t,t+1]$ can only depend on the risk drivers at time $t$ and previous observations do not provide more information. Since we cannot consider numerically any possible function $g$, we assume that the optimal function $g$ belongs to a family of non-linear functions $\mathcal{G}$ which, in this paper, will correspond to a neural network discussed in the next section. }

Moreover, in order to approximate the expectation operator in \eqref{iteratedsystem}, we use a Monte-Carlo sample by generating $M$ random simulations. \textcolor{black}{Given the $\mathbb{P}$-dynamics of the stochastic process $\left\{ Z(t) \right\}_{t=0,\dots,T}$, one can simulate $M$ random observations $\bm Z^{(i)}(t)$ of the random process at time $t$, for $t=0,1, \ldots, T$.} Therefore, the dynamic hedging problem \eqref{iteratedsystem} can be expressed by the following iterative algorithm:
\begingroup
\allowdisplaybreaks
\begin{align}\label{iteratedsystemapprox}
	\bm \theta(t+1)&:=g_{t+1}\left(\bm Z(t)\right) \nonumber\\
	\text{with } g_{t+1}& = \argmin_{g\in\mathcal{G}}\frac{1}{M}\sum_{i=1}^M \left( \rho^{(i)}_{t+1}(S)- g(\bm Z^{(i)}(t)) \cdot \bY^{(i)}(t+1)\right)^2\nonumber\\
	\bm\eta(t+1)&:=h_{t+1}\left(\bm Z(t)\right)\nonumber\\
	\text{with } h_{t+1} & = \argmin_{g\in\mathcal{G}}\frac{1}{M}\sum_{i=1}^M \ell_{\alpha} \left( \rho^{(i)}_{t+1}(S)- \bm\theta(t+1)\cdot \bY^{(i)}(t+1)- g(\bm Z^{(i)}(t)) \cdot \bY^{(i)}(t+1)\right)\nonumber\\
	\rho^{(i)}_{t}(S)&:=\boldsymbol{\theta}^{(i)}(t+1) \cdot \mathbf{Y}^{(i)}(t)+i\ \boldsymbol{\eta}^{(i)}(t+1) \cdot \mathbf{Y}^{(i)}(t).
\end{align}\endgroup
where  \begin{align*}
\boldsymbol{\theta}^{(i)}(t+1)&:= g_{t+1}(\bm Z^{(i)}(t))\\
\boldsymbol{\eta}^{(i)}(t+1)&:= h_{t+1}(\bm Z^{(i)}(t))
\end{align*}
for any $t=T-1,\dots,0$, starting with $\rho^{(i)}_{T}(S)=S^{(i)}$. \textcolor{black}{Hence, we observe that for each time period $[t,t+1]$, the algorithm minimises the aggregate hedging error over all sample paths, taking into account that the hedging strategy should be $\mathcal{F}_t$-measurable.} We note that the algorithm \eqref{iteratedsystemapprox} provides the fair valuation of $S$ at any time $t$ as well as the quadratic and quantile hedging strategies. Indeed, by Lemma \ref{hedgingdecomposition}, the quantile hedging strategy $\bm \xi $ for $\rho_{t+1}(S)$ is the sum of the quadratic hedging strategy $\bm \theta $ for $\rho_{t+1}(S) $ and the quantile hedging strategy $\bm\eta$ for the residual loss. Algorithm \ref{algo1} presents the procedure in its compact form.
\begin{algorithm}[h!]
	\caption{Backward resolution of the dynamic fair valuation problem}
	\begin{algorithmic}[1]
		\State 	$\rho_T \leftarrow S$
		\For{$t=T-1,T-2,...,0$}
		\State $g_{t+1} = \argmin_{g\in\mathcal{G}}\frac{1}{M}\sum_{i=1}^M \left( \rho^{(i)}_{t+1}(S)- g(\bm Z^{(i)}(t)) \cdot \bY^{(i)}(t+1)\right)^2$
		\State $h_{t+1}= \argmin_{g\in\mathcal{G}}\frac{1}{M}\sum_{i=1}^M \ell_{\alpha} \left( \rho^{(i)}_{t+1}(S)- g(\bm Z^{(i)}(t)) \cdot \bY^{(i)}(t+1)\right)$
		\State $\rho^{(i)}_{t}(S) = g_{t+1}\left(\bm Z^{(i)}(t)\right) \mathbf{Y}^{(i)}(t) + i \left( h_{t+1}\left(\bm Z^{(i)}(t)\right)-g_{t+1}\left(\bm Z^{(i)}(t)\right) \right) \mathbf{Y}^{(i)}(t) $
		\EndFor
	\end{algorithmic}\label{algo1}
\end{algorithm}

\subsection{Implementation by neural networks}

In order to implement the Algorithm \ref{algo1}, we need to resort to a non-linear optimisation. In this paper, we implement the algorithm by the use of neural networks (NNs) as these are well suited for this problem. The universal approximation theorem of \cite{hornik1989multilayer} states that networks can approximate any continuous function on a compact support arbitrarily well if we allow for arbitrarily many neurons $q_{1} \in \mathbb{N}$ in the hidden layer. From the universal approximation theorem, we do know that the optimal $g$ function can be approached by a neural network with sufficient layers and neurons. The number of neurons and layers that we need is rather subjective and subject to empirical studies. Here, we follow the work of \cite{fecamp2019risk} and consider three hidden layers of 10 neurons with Relu activation functions. For completeness, we briefly explain the mathematical structure of a neural network in the next paragraph, see \cite{goodfellow2016deep} for more details.

The neural network takes an input of dimension $m$ (the dimension of the risk process $\bm Z$) and outputs a vector of dimension $n+1$ (the number of units invested in the $(n+1)$ assets). The network is characterised by a number of layers $L+1 \in \mathbb{N} \backslash\{1,2\}$ with $m_{l}, l=0, \ldots, L,$ the number of neurons (units or nodes) on each layer: the first layer is the input layer with $m_{0}=m$, the last layer is the output layer with $m_{L}=n+1,$ and the $L-1$ layers between are called hidden layers, where we choose for simplicity the same dimension $m_{l}=p, l=1, \ldots, L-1$. The neural network is then defined as the composition
$$
x \in \mathbb{R}^{m} \mapsto \mathcal{N}(x)=A_{L} \circ \varphi \circ A_{L-1} \circ \ldots \circ \varphi \circ A_{1}(x) \in \mathbb{R}^{n+1}.
$$
Here, $A_{l}, l=1, \ldots, L$, are affine transformations represented by
$$
A_{l}(x)=\mathcal{W}_{l} x+\beta_{l},
$$
for a matrix of weights $\mathcal{W}_{l}$ and a ``bias'' term $\beta_{l}$. The non-linear function $\varphi: \mathbb{R} \rightarrow \mathbb{R}$ is called the activation function and is applied component-wise on the outputs of $A_{l},$ i.e., $\varphi\left(x_{1}, \ldots, x_{p}\right)=\left(\varphi\left(x_{1}\right), \ldots, \varphi\left(x_{p}\right)\right)$. Standard examples of activation functions are the sigmoid, the ReLU, the elu and tanh. 
\section{Numerical example: Portfolio of equity-linked contracts}\label{numericalexamples}

In this section, we determine the multi-period two-step valuation for a portfolio of equity-linked contracts, extending the one-period Example \ref{exampleELoneperiod}. We consider a portfolio of equity-linked life insurance contracts, which guarantee a survival benefit to all policyholders who are still alive at maturity time $T$. The insurance liability can be expressed by
\begin{equation}
S=N(T)\times\max\left(  Y^{(1)}(T),K\right)  , \label{payoff}%
\end{equation}
with $N(t)$ a mortality process counting the number of survivors at time $t$ 
among an initial population of $l_{x}$ policyholders of age $x$, $Y^{(1)}(t)$ a risky asset process, and $K$ a fixed guarantee level. To account for the dependence between financial and actuarial risks, we assume that the dynamics of the stock process and the population force of mortality are given by
\begin{align}
	dY^{(1)}(t)  &  =Y^{(1)}(t)\left(  \mu dt+\sigma dW_{1}(t)\right) \label{ds}\\
	d\lambda_{x}(t)  &  =c\lambda_{x}(t)dt+\eta dW_{2}(t), \label{dlambda}%
\end{align}
with $c,\eta,\mu$ and $\sigma$  positive constants,  $W_{1}(t)=\delta
W_{2}(t)+\sqrt{1-\delta^{2}}X(t)$, and $W_{2}(t)$ and $X(t)$ are independent
standard Brownian motions and \textcolor{black}{$-1\leq \delta \leq 1$} represents the dependence between the stock and the force of mortality. We note that the stochastic force of mortality $\lambda_{x}(t)$ represents the systematic mortality risk, namely the risk that the whole population lives more (or less) than expected. The Ornstein--Uhlenbeck process without mean reversion \eqref{dlambda} was  among others considered in \cite{luciano2008bis} and \cite{luciano2017single}  and was found to provide an appropriate fit to cohort life tables.

The survival function is then defined by
\[
S_{x}(t):=\mathbb{P}\left(  T_{x}>t\right)  =\exp\left(  -\int_{x}%
^{x+t}\lambda_{x}(s)ds\right)  \text{,}%
\]
where $T_{x}$ is the remaining lifetime of an individual who is aged $x$ at
time $0$.

Moreover, to express the unsystematic mortality risk and pooling effects, deaths of individuals are assumed to be independent events
conditional on the population mortality. Further, if we denote $D(t+1)$ the number of deaths
during year $t+1$, the dynamics of the number of active contracts can be
described as a nested binomial process as follows: $N(t+1)=N(t)-D(t+1)$ with
$D(t+1)|N(t),q_{x+t}\sim Bin(N(t),q_{x+t})$. Here, $q_{x+t}$ represents the
one-year death probability%
\[
q_{x+t}:=\mathbb{P}\left(  T_{x}\leq t+1|T_{x}>t\right)  =1-\frac{S_{x}%
	(t+1)}{S_{x}(t)},\text{ for }t=0,\ldots,T-1.
\]
Knowing the dynamics of $N(t)$ and $Y^{(1)}(t)$, we can simulate $M$
scenarios for the mortality and the equity risk factors for $t=1,\ldots,T$. It is clear that in this example, the observable processes of interest at time $t$ are the stock price and the number of survivals: $\bm Z(t):=(Y_{1}(t),N(t))$. Therefore, we have two neural networks: 
\begin{align}
	g_{t+1}: \mathbb{R}^2 \to \mathbb{R}^2, &\ (Y_{1}(t),N(t)) \mapsto  g_{t+1}(Y_{1}(t),N(t))=\bm \theta(t+1),\label{thetaexp}\\
	h_{t+1}: \mathbb{R}^2 \to \mathbb{R}^2, &\ (Y_{1}(t),N(t)) \mapsto  h_{t+1}(Y_{1}(t),N(t))=\bm \xi(t+1),\label{xiexp}
\end{align}
corresponding to the quadratic and quantile hedging strategies for the portfolio of equity-linked contracts, respectively. 

Hereafter, we provide a numerical analysis for the fair dynamic valuation of the insurance liability $S$ introduced above. Our numerical results are obtained by generating $200000$ sample paths for $N(t)$ and $Y^{(1)}(t),$ for $t=1, \ldots, T .$ The benchmark parameters for the financial market are $r=0.01, \mu=0.02$ $\sigma=0.1, K=1,\delta=-0.5$ and $Y^{(1)}(0)=1 .$ The mortality parameters $\left(\lambda_{x}(0)=0.0087, c=0.0750, \xi=0.000597\right)$ follow from \cite{luciano2017single} and correspond to UK male individuals who are aged 55 at time $0$. We assume that there are $l_{x}=1000$ initial contracts at time 0 with a maturity of $T=10$ years.

Figure \ref{figure2} represents prediction intervals for the evolution of the fair valuations, $\rho_t (S)$, for $t=0,\dots,T-1$, obtained by the NN Algorithm \ref{algo1} (left) along with the final payoff $S$ (right). We observe that, as the maturity of the contract increases, the confidence intervals are wider due to the higher uncertainty. Moreover, we remark that the evolution of the fair value through time is smooth and provides a good match to the final payoff.
\begin{figure}[h]
	\centering
	\includegraphics[width=\linewidth]{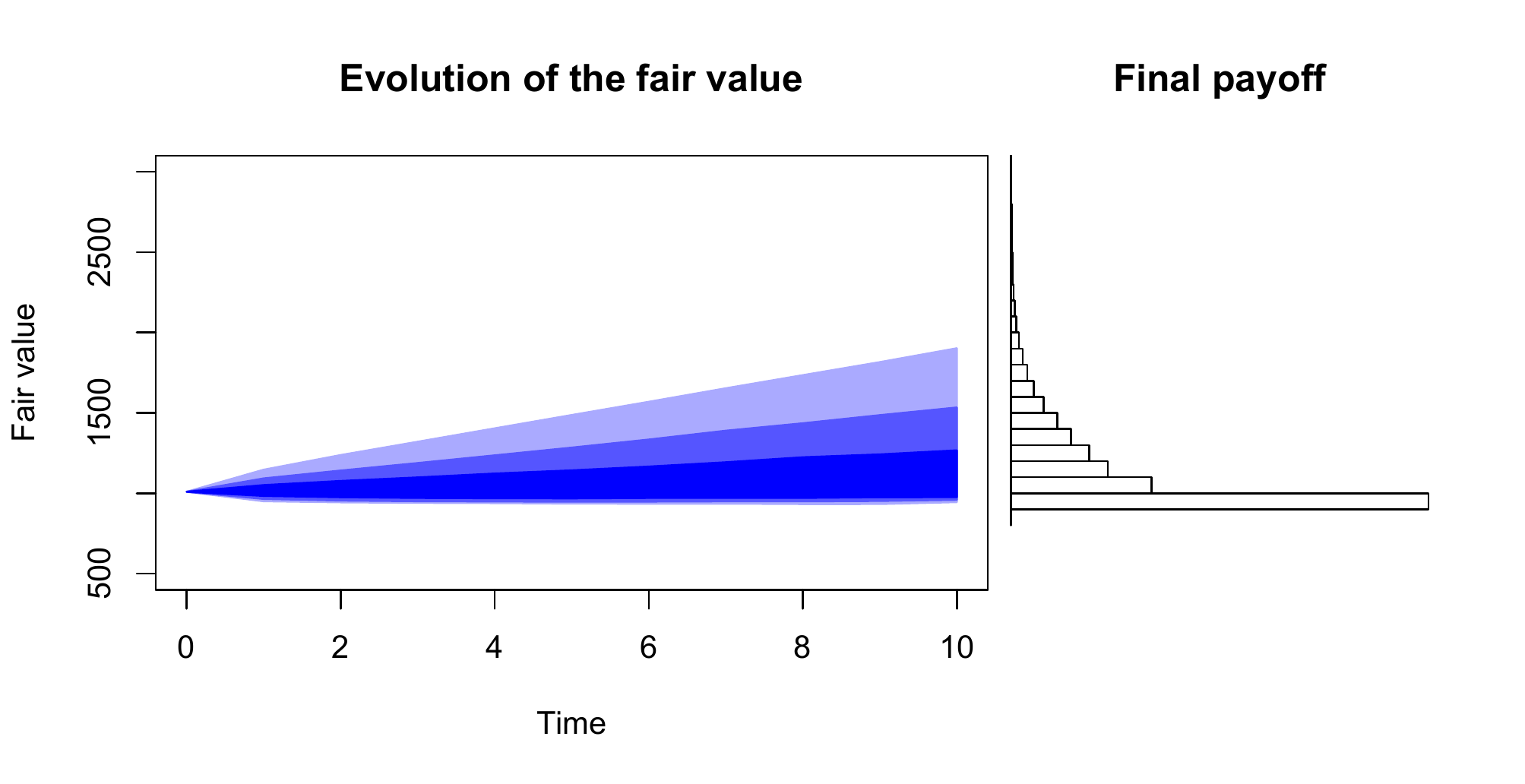}
	\caption{Left: Evolution of the fair valuation from time $0$ to maturity time $T=10$. Right: Histogram of  the final payoff $S=N(T)\times\max\left(  Y^{(1)}(T),K\right)$. Shades in the fan represent prediction intervals at the $50\%$, $80\%$ and $95\%$ level.}
	\label{figure2}
\end{figure}

So far, we did not discuss the rebalancing cost of the quantile hedging portfolio. We recall that, by construction, the hedging portfolio is rebalanced every year in order to satisfy the yearly solvency constraints: 
\begin{equation}\label{varconstraint2}
\text{VaR}_{\alpha,t}\left(\rho_{t+1}(S)-\boldsymbol{\xi}(t+1) \cdot \mathbf{Y}(t+1) \right)=0,\quad \forall t\in \{0,1,\dots,T-1\}.
\end{equation}
Therefore, if e.g. $\alpha=0.95$, there is a probability of at most $0.05$ that the hedging portfolio will not be sufficient to cover the (fair value of) liability. However, there is a priori no guarantee that the payoff of the hedging portfolio will be sufficient to cover the cost of the hedging portfolio for the next year \textcolor{black}{ -- see \eqref{eq:rebal} and the surrounding discussion.} The rebalancing cost at time $t$ is given by
\begin{equation*}
\textcolor{black}{RB}(t)=\boldsymbol{\xi}(t+1) \cdot \mathbf{Y}(t)-\boldsymbol{\xi}(t) \cdot \mathbf{Y}(t),\quad \forall t\in \{1,\dots,T-1\}.
\end{equation*}

Based on our neural network approximation, Figure \ref{figure3} depicts prediction intervals for the rebalancing cost of the hedging portfolio for any $t=1,\dots,T-1$ along with the total rebalancing cost: 
\begin{equation*}
\text{Total rebalancing cost}= \sum_{t=1}^{T-1} e^{-rt}\left(\boldsymbol{\xi}(t+1) \cdot \mathbf{Y}(t)-\boldsymbol{\xi}(t) \cdot \mathbf{Y}(t) \right).
\end{equation*}
First, we notice that intervals of the yearly costs are approximately centered around zero, meaning that there is no yearly rebalancing cost \textit{on average}. Moreover, we can observe that with high probability, the rebalancing cost is lower than $40$, which is approximately $3\%$ of the expected liability. On the aggregate level, Figure \ref{figure3} also shows that that, on average, the rebalancing cost is very low and with high probability ($\alpha=0.95$), the total rebalancing cost will be no more than $10\%$ of the expected liability. Therefore, the cost of the iterated quantile hedging appears to lie at reasonable levels. On Figure \ref{figure4}, we show the final loss of the dynamic hedging strategy. We observe that, with high probability, the hedging portfolio will cover the liability $S$, which follows from the property \eqref{varconstraint2} of quantile hedging. \textcolor{black}{From an economic standpoint, Figure \ref{figure3} illustrates that once the quantile hedging strategy is set up at time 0, we do not expect that significant additional funds need to be raised from investors to set up the quantile hedging strategy for the following time periods.}
\begin{figure}[tbp]
	\centering
	\begin{minipage}{0.4\textwidth}
		\centering
		\includegraphics[width=\textwidth]{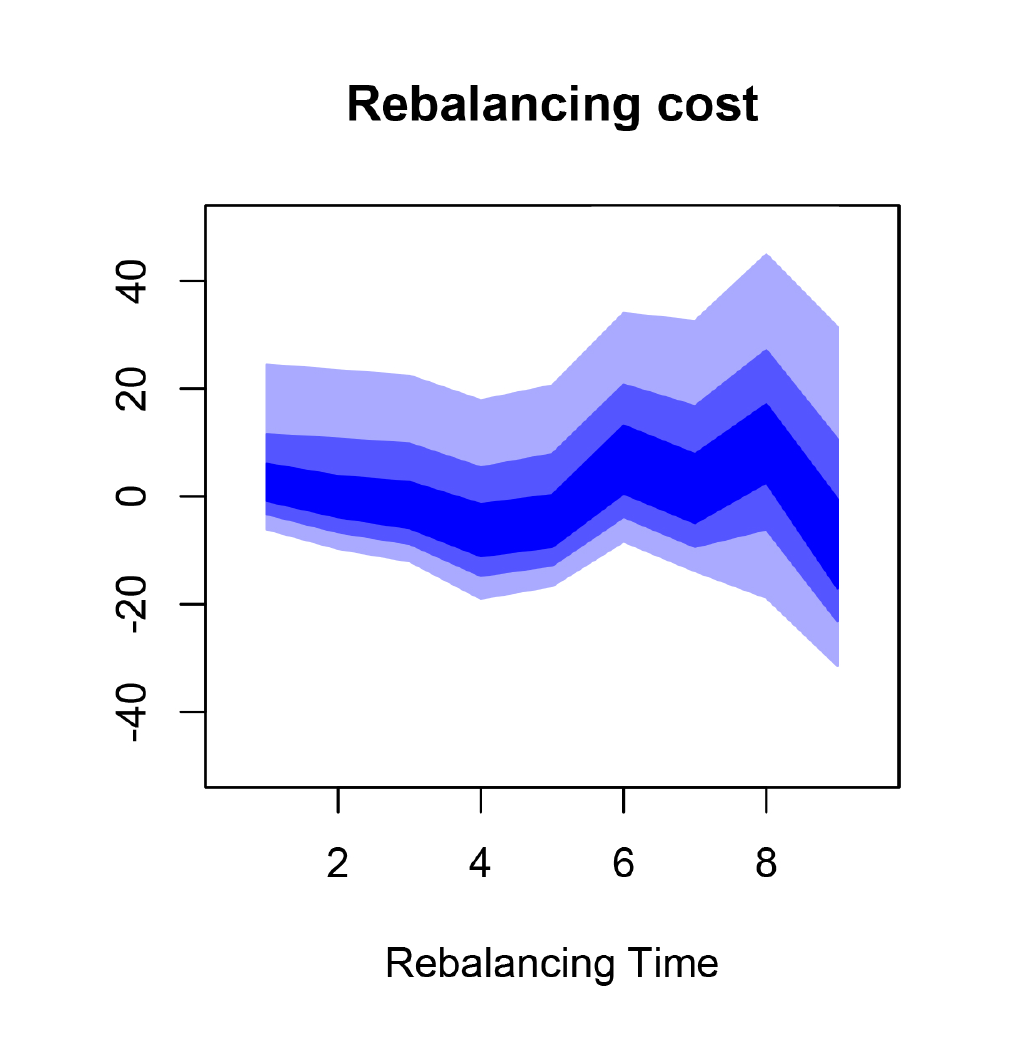} 
	\end{minipage}\hfill
	\begin{minipage}{0.6\textwidth}
		\centering
		\includegraphics[width=\textwidth]{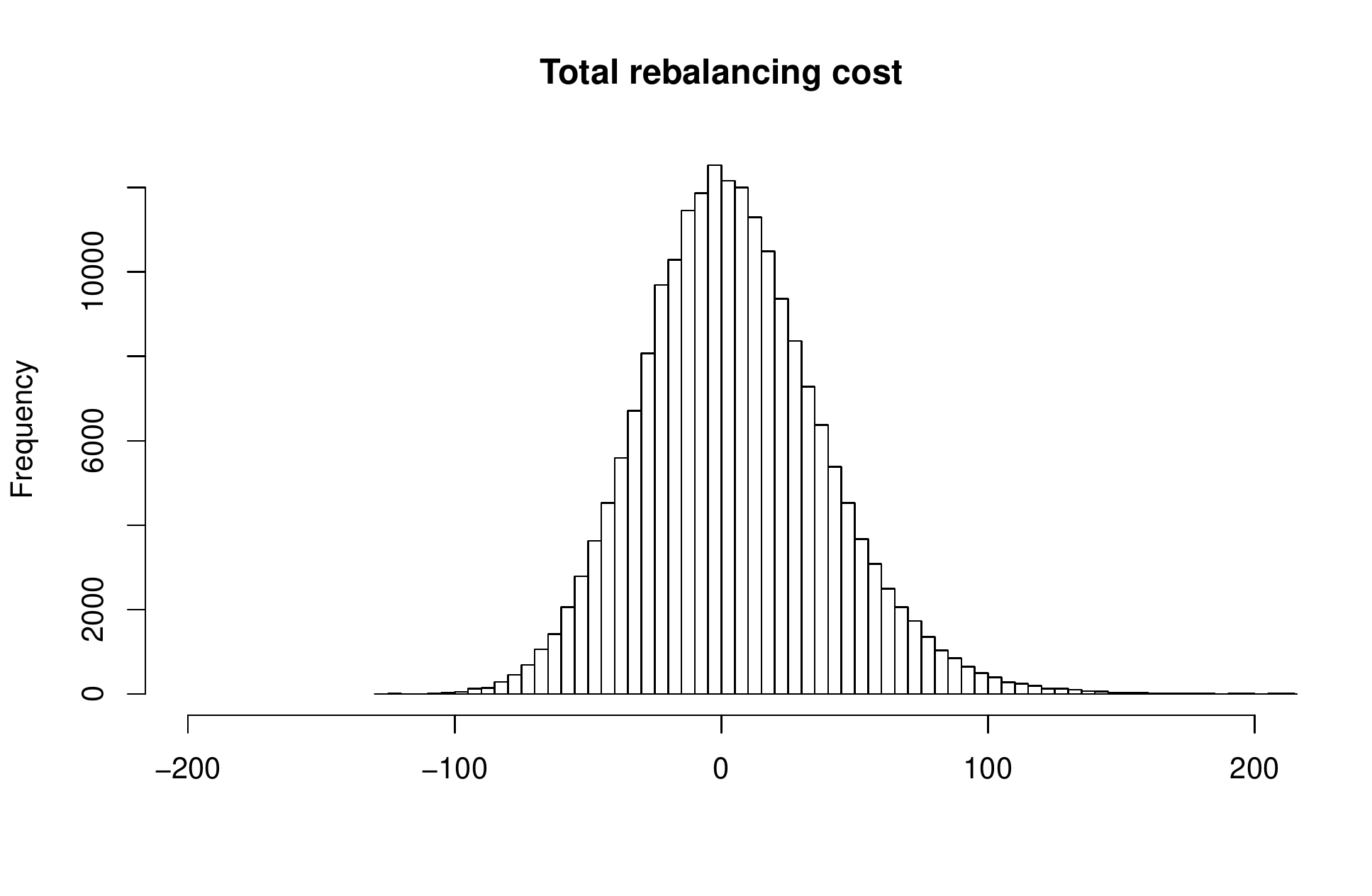}
	\end{minipage}
\caption{Left: Rebalancing cost of the hedging portfolio at any rebalancing times $t=1,\dots,T-1$. Right:  total rebalancing cost. Shades in the fan represent prediction intervals at the $50\%$, $80\%$ and $95\%$ level.}
\label{figure3}
\end{figure}

\begin{figure}[tbp]
	\centering
	\includegraphics[width=0.6\linewidth]{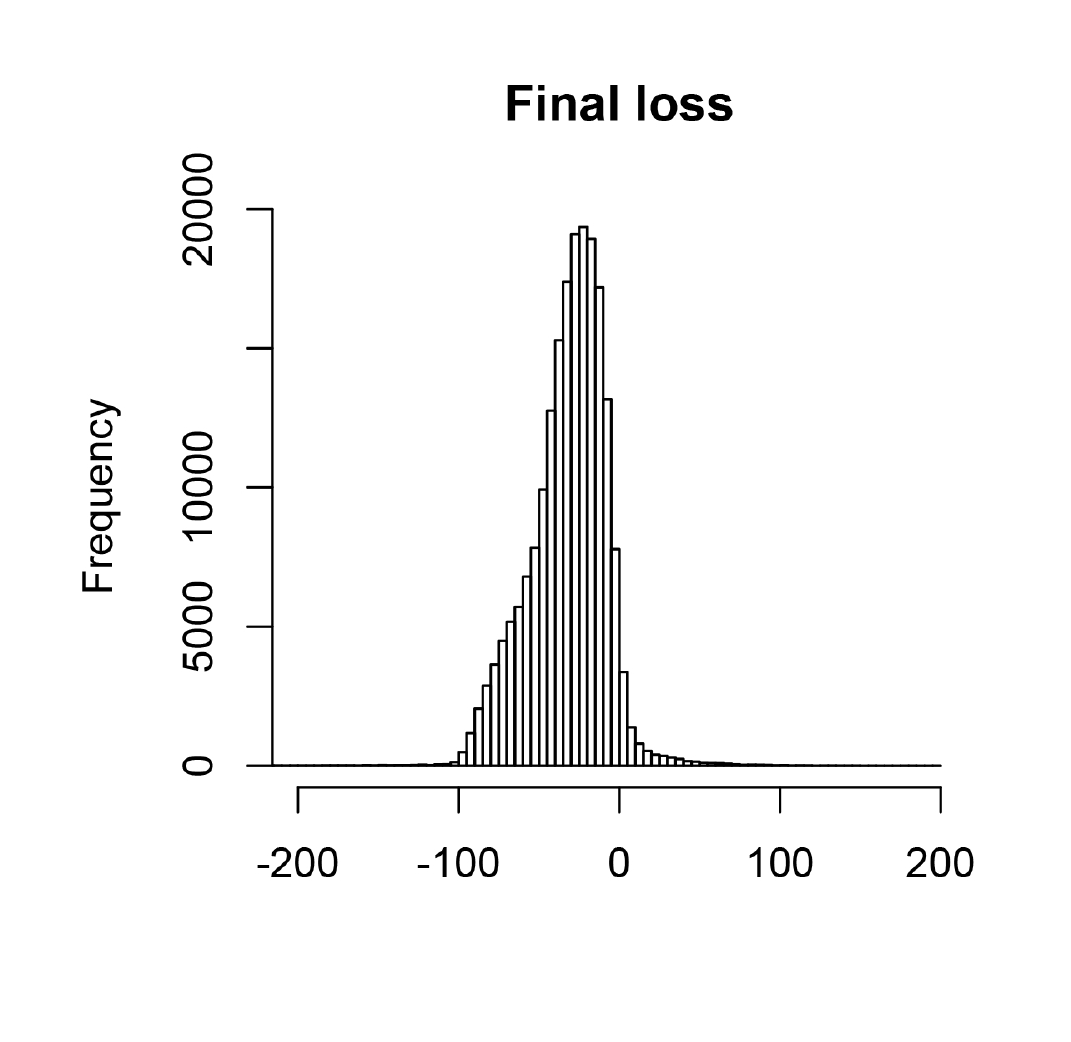}
	\caption{Histogram of the final loss  $S-\boldsymbol{\xi}(T) \cdot \mathbf{Y}(T)$.}
	\label{figure4}
\end{figure}

The neural network estimation allows us to study the non-linearities in the quantile hedging strategy. By expression \eqref{xiexp}, the neural network delivers two outputs for any time $t$, corresponding to the investment in the risk-free asset $Y_{0}(t)$ and $Y_{1}(t)$, respectively. Figure \ref{figure5} represents the number of assets held at time $t=5$ as a function of the stock price at time $5$. As expected, we notice that the investment in the stock is an increasing function of the stock price to better match the terminal liability. On the other hand, due to a compensating effect, the risk-free investment is a decreasing function of the stock price and reaches zero for high stock prices.
\begin{figure}[tbp]
	\centering
	\includegraphics[width=0.8\linewidth]{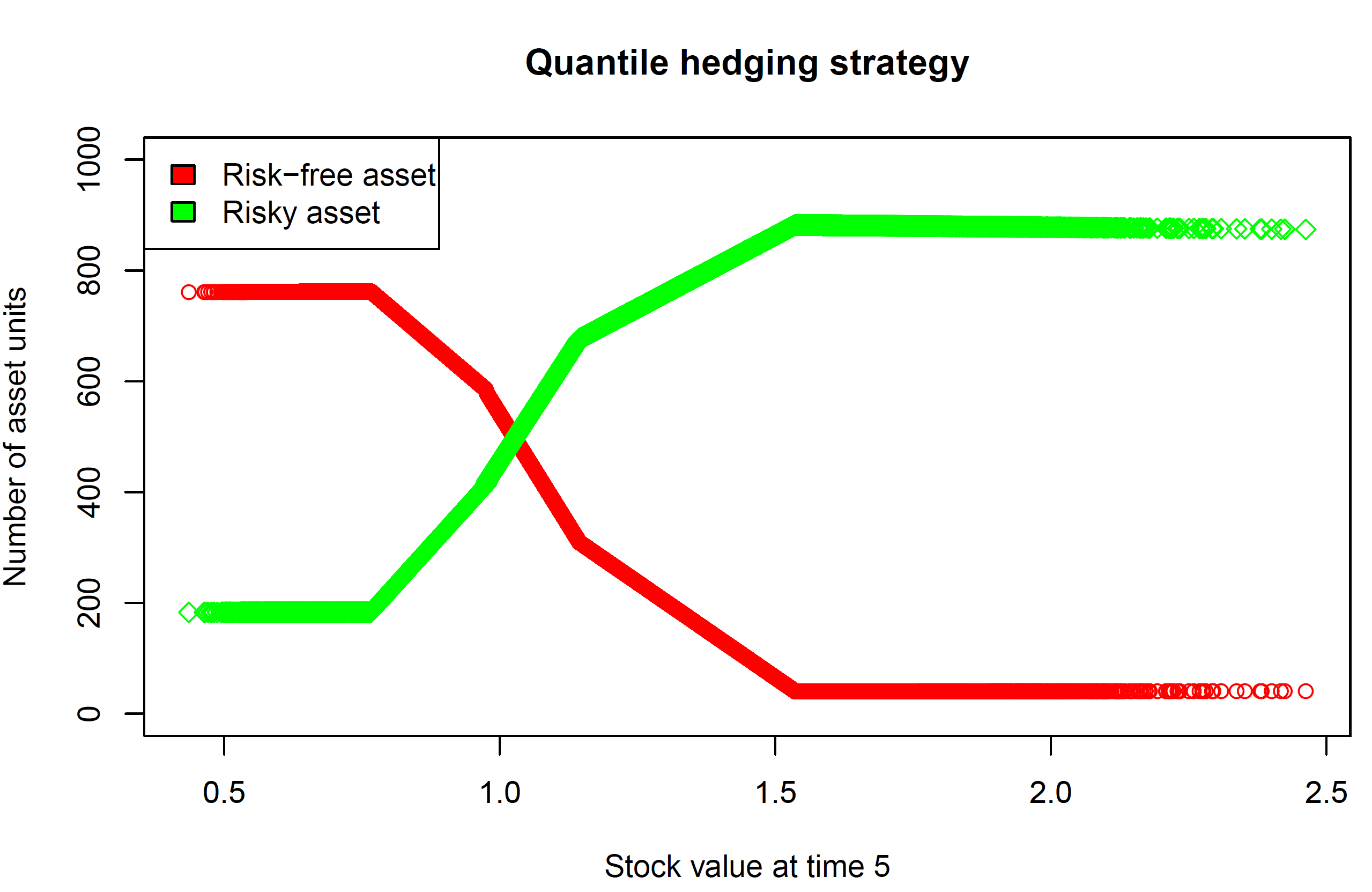}
	\caption{Number of asset units bought at time $t=5$ in the risk-free asset and risky asset under the quantile hedging strategy as function of the asset price $Y^{(1)}(5)$. This strategy corresponds to the expression \eqref{xiexp}: $h_{6}(Y_{1}(5),N(t))=\bm \xi(6)$, with fixed mortality $N(t)=\mathbb{E}[N(5)]$.}
	\label{figure5}
\end{figure}

Finally, we  study key metrics of the quantile hedging residuals in order to assess the accuracy of the neural network algorithm. By \cite{rockafellar2008risk} and \cite{rockafellar2013quadrangle}, it is well-known that the quantile hedging strategy  satisfies the relations:
\begin{align*}
\text{VaR}_{\alpha,t}\left(\rho_{t+1}(S)-\boldsymbol{\xi}(t+1) \cdot \mathbf{Y}(t+1) \right)&=0,\quad \forall t\in \{0,1,\dots,T-1\},\\
 \mathbb{E}\left[\ell_\alpha(\rho_{t+1}(S)-\boldsymbol{\xi}(t+1) \cdot \mathbf{Y}(t+1))\right]&=\dTVaR \left[\rho_{t+1}(S)-\boldsymbol{\xi}(t+1) \cdot \mathbf{Y}(t+1)\right],
\end{align*}
for all $t$ in $\{0,1,\dots,T-1\}$. An appropriate quantile hedging algorithm should therefore have residuals with a VaR close to zero and the average Koenker-Bassett error close to the TVaR deviation. Table \ref{residuals} reports the VaR of the residuals, the average Koenker-Bassett error and the TVaR deviation for all $t$. We observe that the empirical VaR is indeed close to zero and small compared to the expected payoff $ \mathbb{E}\left[S\right]\approx1162$. Moreover, the K-B error is close to the TVaR deviation, hence showing the accuracy of our algorithm. The remaining difference is essentially due to the estimation and simulation error of our approach, which can be further reduced by increasing the simulations or the complexity of the neural network at the cost of higher computational time.
\begin{table}[tbp]
	\centering
	\begin{tabular}{r|r|rr}
		\hline
	Time	& \multicolumn{3}{c}{Residuals} \\
	\cmidrule(r{4pt}){2-4}  & VaR & K-B error & dTVaR \\ 
		\hline
		1 &  0.054 &5.905 &4.243\\ 
		2 & 1.108 &50.284&50.245\\ 
		3 &  -0.031 &29.518&29.518\\ 
		4 & 1.280 & 28.183 &28.076\\ 
		5 &  1.041 &28.307&28.228\\ 
		6 & -2.228 &29.712&29.406\\ 
		7 &  0.966 &29.840&29.763\\ 
		8 &  0.325 &30.685& 30.674\\ 
		9 &  -2.306 &33.829& 33.323\\ 
		10 & -1.462 &  38.812&38.629\\ 
		\hline
	\end{tabular}
\caption{This table reports Value-at-Risk (VaR), Koenker-Bassett error (K-B error) and Tail Value-at-Risk deviation (dTVaR) of quantile hedging residuals at confidence level $\alpha=0.95$.} \label{residuals}
\end{table}

\section{Concluding remarks}\label{conclusion}

We discussed the fair valuation of insurance liabilities in a multi-period discrete-time setting. As insurance liabilities are not directly traded in the financial market, the valuation requires a decomposition into a ``hedgeable part" and a ``residual part". For the first part, it seems that the quadratic objective has become a standard practice probably due to its analytical tractability and the fact that the resulting hedging portfolio targets the expected liability (\citealp{pelsser2016difference}). However, there is still an open debate on how to appropriately treat the residual part and define an appropriate ``risk margin" (see e.g. \citealp{pelkiewicz2020review} for a review on the Solvency II risk margin). In the literature, different approaches were considered to value the residual risk, either by an Esscher valuation operator \citep{deelstra2020valuation}, a standard-deviation principle \citep{barigou2020pricing,ahmad2020,ze2020,delong2019fair},  or a cost-of-capital principle \citep{pelsser2011}. 

Of the above valuation principles, the cost-of-capital approach takes account of the need to set up VaR-neutral portfolio; but in this approach VaR-neutrality becomes divorced from hedging considerations. Rather than using a cost-of-capital principle on the residual risk, we propose to quantile-hedge it. While both approaches lead to a VaR-neutral portfolio, our proposed approach has two noticeable advantages. First, the quantile hedging portfolio is a TVaR deviation risk minimiser and therefore better accounts for the tail risk (see Lemma \ref{lemmatvar}, Examples \ref{example: regulatory arbitrage}, \ref{exampleELoneperiod}). Second, our quadratic-quantile approach shows that the residual risk can still be partially hedged if one switches from a quadratic to a quantile hedging objective. This is especially relevant when there is a non-linear relationship between the insurance liability and traded instruments.

Moreover, we proposed a simulation-based general algorithm for the practical implementation of our approach. In this paper, we focused on a neural network implementation for quantile hedging but the algorithm can be easily adapted for a general loss function $\ell$ and other non-linear optimisers. This paper focused on the multi-period hedging of a cash-flow with maturity time $T$. Finally, this paper did not explicitly account for capital injections and withdrawals from the shareholders’ viewpoint and their option to default as considered e.g. in \cite{engsner2021multiple}. These points are left for future research.

\section{Acknowledgements} 

\textcolor{black}{The authors would like to thank the Editor, three anonymous referees and Jan Dhaene who provided useful and detailed comments that substantially improved the current manuscript.} Karim Barigou acknowledges the financial support of the Joint Research Initiative on ``Mortality Modeling and Surveillance” funded by AXA Research Fund.

\end{document}